\documentclass[a4paper,twocolumn,10pt]{quantumarticle}
\pdfoutput=1
\usepackage[numbers,sort&compress]{natbib}
\usepackage[colorlinks,breaklinks]{hyperref}
\usepackage{amsmath,amssymb,amsthm,bm,dsfont,mathdots}
\usepackage{physics}
\usepackage{graphicx}

\newcommand{\id}{\ensuremath{\mathds{1}}}
\newcommand{\Irr}{\ensuremath{\mathrm{Irr}}}
\DeclareMathOperator{\polylog}{polylog}
\renewcommand{\vec}[1]{\boldsymbol{#1}}

\newtheorem{thm}{Theorem}

\newtheorem{obs}[thm]{Observation}
\newtheorem{lem}[thm]{Lemma}

\begin{document}
\nonfrenchspacing

\title{Phases of matrix-product states with symmetries and measurements: Finite nilpotent groups}

\author{David Gunn}
\affiliation{TUM School of Natural Sciences, Technical University of Munich, James-Franck-Str. 1, D-85748 Garching, Germany}
\affiliation{Munich Center for Quantum Science and Technology (MCQST), Schellingstr. 4, 80799 München, Germany}
\affiliation{International Iberian Nanotechnology Laboratory (INL), Av. Mestre José Veiga, 4715-330 Braga, Portugal}
\author{Georgios Styliaris}
\affiliation{Max Planck Institute of Quantum Optics, Hans-Kopfermann-Str. 1, Garching 85748, Germany}
\affiliation{Munich Center for Quantum Science and Technology (MCQST), Schellingstr. 4, 80799 München, Germany}
\author{Barbara Kraus}
\affiliation{TUM School of Natural Sciences, Technical University of Munich, James-Franck-Str. 1, D-85748 Garching, Germany}
\affiliation{Munich Center for Quantum Science and Technology (MCQST), Schellingstr. 4, 80799 München, Germany}
\author{Tristan Kraft}
\affiliation{TUM School of Natural Sciences, Technical University of Munich, James-Franck-Str. 1, D-85748 Garching, Germany}
\affiliation{Munich Center for Quantum Science and Technology (MCQST), Schellingstr. 4, 80799 München, Germany}

\begin{abstract}
We study phases of one-dimensional matrix-product states (MPS) when transformations are restricted to symmetric local circuits supplemented with symmetric measurements and feedforward ($G$-CMF). Building on the framework introduced in~\href{https://doi.org/10.1103/PhysRevB.111.115110}{Gunn~\textit{et al.},~Phys.~Rev.~B~\textbf{111},~115110~(2025)}, we extend the analysis to all finite nilpotent groups for which we obtain a complete classification of $G$-CMF phases. We construct explicit symmetry‑respecting protocols that map any symmetry-protected topological (SPT) or non-normal (GHZ‑type) MPS to the trivial phase---and vice versa---with success probability approaching one in the thermodynamic limit. The key technical ingredient is a finite hierarchical structure of irreducible representations of nilpotent groups, which enables successive rounds of symmetric measurements to systematically reduce non‑abelian components to abelian ones. Our results demonstrate that allowing symmetric measurements and feedforward fundamentally simplifies the phase structure of 1D systems with nilpotent symmetries: all SPT and non-normal MPS phases collapse into a single asymptotically equivalent phase under $G$-CMF transformations.

\end{abstract}

\maketitle

\section{Introduction}\label{sec:Intro}

The classification of quantum phases of matter is a central problem in many-body physics. As in many areas of theoretical physics and mathematics, symmetries serve as a fundamental organizing principle. For the case of gapped phases in one spatial dimension, this perspective has led to a rich structure of phases beyond the traditional Landau paradigm, such as symmetry-protected topological (SPT) phases. These phases can be systematically classified within the framework of tensor networks, particularly matrix-product states (MPS)~\cite{cirac2021matrix}. Within this formalism, it is well-established that the phases of MPS protected by global on-site symmetries correspond to elements of the second cohomology group of the symmetry group with coefficients in $U(1)$~\cite{Chen2011_Phases1,Schuch2011_Phases2}. This classification reflects the interplay between the symmetry group and the structure of the tensors that define the MPS.

An alternative yet equivalent perspective on phases has emerged, based on the action of quantum circuits. In this formalism, two translation-invariant MPS (parameterized by system size $N$) are said to belong to the same phase if there exists a local quantum circuit with depth $O(\polylog N)$ that maps one to the other~\cite{chen2010local,hastings2005quasiadiabatic,nachtergaele2019quasi,coser2019classification}. When symmetries are present, each gate in the circuit is required to commute with the given representation of the symmetry group. This reformulation has proven powerful in connecting notions of phase equivalence with computational complexity and operationally meaningful transformations.

A practically motivated extension of this framework, inspired by earlier ideas~\cite{raussendorf2001one,raussendorf2005long,aguado2008creation,meignant2019distributing,watts2019exponential}, is to include, along with quantum circuits, quantum measurements and feedforward~\cite{piroli2021quantum}. This refers to the case in which unitary gates can be classically conditioned on outcomes of previous measurements in parts of the circuit. This generalized setting is especially natural from the viewpoint of entanglement theory, which is based on the paradigm of local operations and classical communication (LOCC), which includes measurements and feedforward~\cite{chitambar2014everything}. In this context, a natural question emerges: How does the classification of phases change when measurements and feedforward are allowed, with or without symmetry constraints?

Several results have been obtained in recent years addressing this question without symmetry constraints, particularly for 1D systems~\cite{piroli2021quantum,lootens2025lowdepth,stephen2024preparing,smith2024constant,sahay2025classifying,malz2024preparation} and for topological phases in 2D~\cite{tantivasadakarn2021long,bravyi2022adaptive,lu2022measurement,tantivasadakarn2023hierarchy,li2023symmetry,zhang2024characterizing,ren2025efficient}. In most of these cases, the focus is on how the addition of measurements and feedforward expands the landscape of possible transformations, including experimental implementations~\cite{foss2023experimental,iqbal2024non,chen2025nishimori}.

When symmetries are imposed, it is essential that \emph{all} operations---including measurement, not only the unitary gates---respect the symmetry. To address this, Ref.~\cite{Gunn2025} introduced a consistent symmetry-respecting framework called \emph{G-symmetric Circuits with Measurements and Feedforward} ($G$-CMF). Within this setting, every quantum operations must be symmetric with respect to a fixed on-site representation of a symmetry group $G$. In particular, Ref.~\cite{Gunn2025} established the following three key results: (i) For a finite abelian symmetry group $G$, all symmetric normal MPS collapse into a single $G$-CMF phase. This is in contrast to the unitary-only case, where the classification via cohomology classes distinguishes multiple SPT phases. (ii) The same $G$-CMF phase also contains all non-normal (i.e., long-range entangled) MPS. In the unitary-only case, these states belong to a distinct phase even in the absence of symmetry protection. (iii) The $G$-CMF classification for normal MPS also collapses to a single phase for class-2 nilpotent groups, a family that includes all finite abelian groups but also extends beyond them. Together, these results demonstrate a dramatic simplification of phases enabled by symmetric measurements and feedforward.

Here, we considerably extend these results by providing a systematic classification of $G$-CMF phases for \emph{all} finite nilpotent groups. Nilpotent groups have the special property that the adjoint action, ${\rm ad}_g(h)=[g,h]$, where $[g,h]=g^{-1}h^{-1}gh$ denotes the group commutator, becomes trivial after finitely many iterations: for every $g,h\in G$ one has $({\rm ad}_g)^M(h)=e$ for some finite integer $M$. The integer $M$ is the \emph{nilpotency class} of $G$. Although nilpotent groups are typically non-abelian---except in the case of class-1 nilpotent groups, which are precisely the abelian groups---the converse does, obviously, not hold; for example, the symmetric group $S_3$ is non-abelian but not nilpotent.

Below, we construct explicit $G$-CMF protocols that transforms a distinguished representative of any SPT phase into the trivial phase---and vice versa---using a finite number of rounds of symmetric measurements, where the number of rounds is solely determined by the nilpotency class of $G$. Although our transformations are not exact, we show that they succeed with probability approaching one in the thermodynamic limit $N\rightarrow\infty$. This establishes that all SPT phases of nilpotent groups become trivial once symmetric measurements and feedforward are allowed. We then consider states with long-range correlations, i.e., non-normal MPS such as GHZ-type states. For these, we present a protocol---combining the protocol used for SPT phases with an additional fusion measurement step---that transforms the trivial phase into non-normal GHZ states. Together with the reverse transformation, constructed analogously to the SPT case, this demonstrates that such non-normal phases also collapse to the trivial phase under $G$-CMF transformations.

\begin{figure*}
    \centering
    \includegraphics[width=0.85\linewidth]{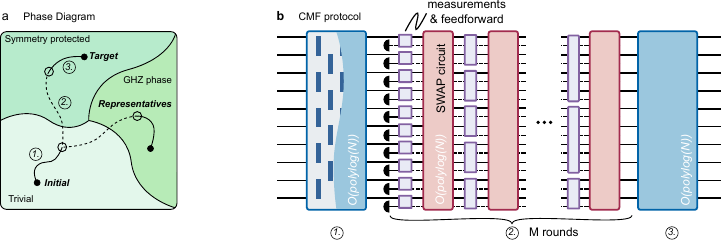}
    \caption{(a) We consider transformations between states belonging to different phases. Starting from an initial state (black dot), for instance in the trivial phase, we first transform the state into a phase representative (black circle). This transformation can be implemented by a local symmetric circuit of depth $O(\polylog N)$~\cite{chen2010local,hastings2005quasiadiabatic,nachtergaele2019quasi,coser2019classification}. In the second step, as we will show in detail below, the phase representative is transformed into representatives of other phases using CMF transformations, which include symmetric measurements. In the final step, we again apply a local symmetric circuit of depth $O(\polylog N)$ to obtain the desired target state. By constructing explicit protocols that transform between phase representatives, we thereby demonstrate a trivialization of the phase diagram. (b) A circuit representation of a CMF transformation is shown schematically. Blue boxes represent local symmetric unitary circuits corresponding to transformations within phases, i.e., the solid lines in panel (a). Red boxes represent CMF layers, each consisting of multiple rounds of symmetric measurements followed by short-depth symmetric quantum circuits that may depend on the measurement outcomes. For class-$M$ nilpotent groups, our protocols require exactly $M$ such rounds. The associated circuits are composed of sequences of nearest-neighbor SWAP gates. While some of our protocols make use of auxiliary systems, other achieve the desired transformation without auxiliary systems.}
    \label{fig:placeholder}
\end{figure*}

As we consider non-abelian groups, the mathematical structure underlying our analysis becomes richer. At the center of our results is the fact that the irreducible representations of nilpotent groups can be organized in a finite hierarchy (see Secs.~\ref{Sec:OutlineOfTheProtocol}, and~\ref{sec:NilpotentGroups}, and Refs.~\cite{Gelaki2008_NilpotentFusionCategories,lootens2025lowdepth}). This hierarchy allows us to design protocols that repeatedly measure and block physical sites, by rules derived from representation theory, in such a way that after each round of the protocol the resulting state descends to a lower level of the hierarchy. Since nilpotent groups have a finite hierarchy, the procedure necessarily terminates after finitely many rounds, ultimately reaching the lowest level. At this stage the state can be transformed into the desired target state via a single round of measurements.

The paper is organized as follows. Section~\ref{sec:prelim} provides the necessary mathematical background on phases of MPS. Section~\ref{Sec:OutlineOfTheProtocol} then outlines the key idea of the protocol that allows us to show that all phases trivialize. Section~\ref{sec:NilpotentGroups} provides a more detailed analysis of the mathematical structure of nilpotent groups, including the origin of the finite hierarchy of irreducible representations. Sections~\ref{sec:ProtocolSPT} and~\ref{sec:GHZProtocol} then establish that SPT and non-normal phases trivialize under $G$-CMF for all finite nilpotent groups.

\section{Preliminaries}\label{sec:prelim}

A tensor $A$ defines a translationally invariant MPS with periodic boundary conditions on $N$ sites via 
\begin{equation}
\ket{\psi_N[A]}=\sum_{i_1, \dots,i_N} \mathrm{tr}[A^{i_1}\cdots A^{i_N}]\ket{i_1\dots i_N}.
\end{equation}
Here, $A^i$ is a $D\times D$ matrix for all $i\in\{0,\dots,d-1\}$, where $d\in\mathbb{N}$ is the physical dimension and $D\in\mathbb{N}$ is the \textit{bond dimension}. Mapping the tensor $A^i$ to the product $\tilde{A}^{i_1,\dots,i_l}=\prod_{j=1}^l A^{i_j}$ corresponds to ``blocking" $l$ physical sites together into a supersite. Any MPS tensor, after blocking sufficiently many sites, can be brought into a canonical form, $A^i = \bigoplus_{j} A_j^i$, with the property that the $A^i$ span all matrices with the same block structure~\cite{Cirac2017_FundThm1}. An MPS is called \emph{normal} if, after blocking, its tensor in canonical form  has only one block. Given a tensor A, one may also define its so-called \emph{fiducial state} as 
\begin{equation}\label{eq:fiducial}
    \ket{A}=\sum_{i,l,j}(A^i)_{lj}\ket{l}\otimes\ket{i}\otimes\ket{j},
\end{equation}
where the system in the middle represents the physical site, whereas systems on the right and left represent virtual sites.

Any MPS tensor in its canonical form can be associated with a parent Hamiltonian for which the MPS is a ground state~\cite{Fannes1992,Schuch2011_Phases2,PerezGarciaEtAl2007_MPSrepresentations}. Then, two systems are said to be in the same phase if their parent Hamiltonians can be continuously connected by a path of local Hamiltonians without closing the spectral gap above the ground-state subspace in the thermodynamic limit. An alternative yet equivalent definition can be stated in terms of quantum circuits. Here, two MPS belong to the same phase if they can be transformed into one another with ``short-depth'' local circuits~\cite{chen2010local}. More precisely, two MPS, $\ket{\psi_N[A]}$ and $\ket{\psi_N[B]}$, belong to the same phase if there exists a sequence of unitaries, $U_N$, such that (i) $\| \ket{\psi_N[A]} - U_N \ket{\psi_N[B]}\|$ converges to $0$ as $N \to \infty$, and (ii) each $U_N$ can be decomposed into a $O(\polylog N)$-depth quantum circuit consisting of $O(1)$-local nearest neighbor gates. Refs.~\cite{Chen2011_Phases1,Schuch2011_Phases2} showed that two MPS belong to the same phase if and only if they have the same number of blocks in canonical form.

In addition to this circuit-based definition of phases, one can consider systems which are constrained to respect a given symmetry. In this case, new phases of matter emerge: so-called \emph{symmetry protected topological} (SPT) phases. In 1D, these phases were characterized in Refs.~\cite{pollmann2010entanglement, Chen2011_Phases1,Schuch2011_Phases2}. In this work, we consider finite global on-site symmetries of the form $U_g^{\otimes N}$, where $U_g$ with $g\in G$, is a unitary linear representation of a finite group $G$, and $N$ is the number of sites. A symmetric MPS then satisfies
\begin{equation}
    U_g^{\otimes N}\ket{\psi_N[A]}\propto \ket{\psi_N[A]}.
\end{equation}
The circuits defining the phases are then restricted such that each local gate acting on $n$ sites must commute with the symmetry $U_g^{\otimes n}$. Phases are then characterized by how the tensor $A$ transforms under the action of the symmetry~\cite{pollmann2010entanglement,Chen2011_Phases1,Schuch2011_Phases2}.

For normal MPS, after blocking and bringing it to canonical form, the symmetry condition can be expressed at the virtual level as~\cite{SanzEtAl2009_MPSHamiltonians, Schuch2011_Phases2}
\begin{equation}\label{eq:injectivesym}
    \sum_j (U_g)_{ij} A^j = e^{i\phi_g} \omega_{g}^\dagger A^i \omega_{g}.
\end{equation}
Here, the phases $e^{i\phi_g}$ form a 1D unitary irreducible representation (irrep) of $G$. The $\omega_g$ are referred to as \emph{virtual symmetries} and form a \emph{projective representation} of $G$, meaning that 
\begin{equation}
    \omega_g\omega_h= \gamma(g,h) \omega_{gh},
\end{equation}
for all $g,h\in G$, where $\gamma: G\times G \rightarrow U(1)$ is a \emph{cocycle}. As any transformation $\omega_g\mapsto \nu(g) \omega_g$ leaves Eq.~\eqref{eq:injectivesym} invariant for any $\nu_g\in U(1)$, the cocycle $\gamma (g,h)$ is defined only up to the equivalence relation
\begin{equation}
    \gamma (g,h)\sim \frac{\nu(gh)}{\nu(g) \nu(h)} \gamma (g,h).
\end{equation}
The resulting equivalence classes can be shown to be isomorphic to the second cohomology group $H^2(G,U(1))$. We label these cohomology classes by $\mu$, with $\mu=0$ denoting the trivial class corresponding to linear representations. The phase of a normal MPS with respect to a given symmetry is fully specified by the cohomology class of the corresponding projective representation~\cite{pollmann2010entanglement, Chen2011_Phases1,Schuch2011_Phases2}.

For non-normal MPS, phases are specified by a subgroup $H\le G$ which determines a permutation of the blocks in the canonical form, and by a cohomology class $\mu\in H^2(H,U(1))$ of $H$~\cite{Chen2011_Phases1, Schuch2011_Phases2}. That is, phases are specified by a tuple $(H,\mu)$, with SPT phases corresponding to $(G,\mu\ne 0)$. The simplest case of non-normal MPS corresponds to $(H=1,\mu=0)$, for which the symmetry condition takes the form
\begin{equation}\label{eq:noninjectivesym}
    \sum_j (U_g)_{ij}\,\Bigl( \bigoplus_k A^j_k \Bigr)
    = P_g^\dagger \,\Bigl( \bigoplus_k e^{i\phi_g^k} A^i_k \Bigr)\, P_g,
\end{equation}
where $ e^{i\phi_g^k}$ are 1D irreps and the virtual symmetries $P_g$ are permutation representations of the subgroup $H$ permuting the blocks~\cite{Schuch2011_Phases2}.

Note that the different phases are labeled by group properties~\cite{note1} and are, in this sense, independent of the particular physical on-site representation $U_g$. However, for a fixed on-site symmetry, there may be no corresponding symmetric MPS that belongs to a given phase; see Ref.~\cite{Gunn2025} for an example. To bypass this technicality, we assume throughout this paper that the physical on-site symmetry is given by the regular representation,
\begin{equation}
U_g^{{\rm reg}}=\sum_{h \in G}\ket{g\circ h}\bra{h},
\end{equation}
where $\{\ket{h}\}_{h\in G}$ is an orthonormal basis and $\circ$ the group multiplication in $G$. The regular representation of a finite group contains every irrep of $G$, with multiplicity equal to its dimension, and moreover any finite-dimensional representation of $G$ appears as a subrepresentation of $(U_g^{\mathrm{reg}})^{\otimes k}$ for some finite $k$. Consequently, for any desired effective symmetry $U_g^{{\rm eff}}$ there exists a finite blocking-length---independent of system-size---at which the physical symmetry acts as $U_g^{{\rm eff}}$ on a local subspace. As a result, one can show that, with $U_g^{{\rm reg}}$ as the physical on-site symmetry, for every phase $(H,\mu)$ there is a state that belongs to that phase~\cite{Gunn2025}.

One can also consider phases of MPS with respect to short-depth circuits and \emph{measurements}~\cite{piroli2021quantum}. In this framework, one may introduce $O(1)$ auxiliary systems on-site, entangle them with the system via short-depth circuits, and measure them. Crucially, the circuits may also depend on the outcomes of previous measurement outcomes. These combined operations are referred to as \textit{circuits, measurements and feedforward} (CMF). Naturally, when considering phases with respect to CMF, what is relevant is the cumulative circuit depth of the protocol.

In this work, we restrict ourselves to projective measurements, i.e., to collections $\{P_i\}_i$ with $P_i^2=P_i\ge0$ and $ \sum_i P_i=1$. Moreover, we consider \emph{asymptotically deterministic transformations}; that is, we do not allow post-selection, and we require the success probability of the protocol to converge to one as $N \rightarrow \infty$. Two MPS are therefore said to belong to the same CMF phase if they can be asymptotically transformed in one another---in the sense of convergent sequences above---via CMF protocols of cumulative circuit-depth $O(\polylog N)$. In Ref.~\cite{piroli2021quantum}, it was shown that the inclusion of measurements and feedforward collapses the entire phase diagram of MPS into a single phase.

As in the case without measurements, one may also restrict CMF operations by imposing a global symmetry. This is precisely the setting considered in this work. As CMF includes quantum circuits as a subset of its allowed operations, the gates appearing in these circuits must satisfy the usual symmetry constraint: each local gate must commute with the symmetry. In addition, one imposes that the measurement operator commute with the symmetry. That is, for each on-site measurement $\{P_i\}_i$ acting on $n$ sites, we require $[P_i,U_g^{\otimes n}]=0$ for all $i$. One also imposes that auxiliary systems may be added and discarded only on-site and only if they are in a symmetric pure state, unentangled with the rest of the system. Since the discarded and added auxiliary systems need not be the same, one can implement gates that commute up to a phase~\cite{Gunn2025}. We refer to this CMF constrained by symmetries as $G$-CMF, where $G$ refers to a given symmetry group~\cite{Gunn2025}. Note, that $G$-CMF without measurements coincides with symmetric circuits, and it coincides with CMF when the symmetry constraints are not imposed (see also Ref.~\cite{Gunn2025} for a more in-depth discussion). As a result, any phase diagram derived from $G$-CMF operations will be a coarse-graining of the standard phase diagram with respect to symmetric circuits and a fine-graining of the phase diagram with respect to CMF.

Under these conditions, we showed in Ref.~\cite{Gunn2025}, that the phase diagram of MPS collapses to a single phase for finite \textit{abelian} symmetry groups. The proof proceeds by selecting a representative for each phase and explicitly constructing symmetric CMF transformations between any pair of such representatives. Since symmetric CMF includes symmetric quantum circuits as a subset of its operations, the existence of these transformations implies that the entire phase diagram collapses; once a representative of a phase can be reached, standard symmetric circuits allow one to reach any other state within that phase.

A key ingredient in the abelian case is the ability to perform symmetric measurements with \emph{rank‑1} projectors. Because all irreducible representations of a finite abelian group are one-dimensional, the on-site symmetry decomposes into 1D invariant subspaces, allowing complete symmetric projective measurements consisting entirely of rank‑1 projectors. This is demonstrated most clearly in considering transformations to the trivial phase. Taking a symmetric product state as a representative of the trivial phase, one can simply perform symmetric rank‑1 measurements on each site; then, clearly, one can transform any state, in any phase, to the trivial phase.

For abelian symmetries, such a measurement is possible. In contrast, for non-abelian symmetries, where at least one higher-dimensional irrep exists, the condition $[P_i, U_g] = 0$ for all $g \in G$ may not restrict the measurement operators $P_i$ to rank-1 projectors. In fact, for a symmetry $U_g$ containing a higher-dimensional irrep, there exists no complete projective measurement consisting only of rank-1 projectors. Thus, a direct transformation to a trivial state through a single round of on-site measurements does not exist in general. However, if a state is locally supported only on invariant subspaces carrying commutative (1D) irreps, then one may, for that given state, deterministically obtain outcomes corresponding to rank-1 projectors. Motivated by this, Ref.~\cite{Gunn2025} showed that phases with respect to the non-abelian dihedral group of eight elements, $D_8$ (the symmetry group of the square) trivialize under $G$-CMF.

The key structural ingredient there was the representation theory of $D_8$. This group has five linear irreps. Let us label them $\{1,2,3,4,5\}$, where $1$ denotes the trivial representation, $2,3$ and $4$ are commutative irreps, and $5$ is a non-commutative irrep. Crucially, the tensor product of two copies of the non-commutative irrep decomposes fully into commutative irreps; namely,
\begin{equation}
5\otimes 5 \cong 1\oplus 2 \oplus 3 \oplus 4.
\end{equation}
This implies that, after performing local projections onto each irrep $P_{\alpha \in \{1,2,3,4,5\}}$, any outcome corresponding to the non‑commutative irrep $P_{\alpha = 5}$ can be combined with another such outcome to form a composite \emph{supersite}~\cite{note2}. The resulting supersite has support only on subspaces corresponding to one-dimensional irreps. A subsequent symmetric measurement on this supersite then deterministically yields an outcome corresponding to a rank-1 projector. In this way, Ref.~\cite{Gunn2025} demonstrated that SPT phases ($H=G$, $\mu=1)$ and the non-normal phase ($H=1, \mu=0)$ of $D_8$ trivialize. In what follows, we extend these protocols beyond $D_8$ to all nilpotent groups. 

\section{Outline of the protocol}\label{Sec:OutlineOfTheProtocol}

Here, we outline how the above example can be generalized to all nilpotent groups. The key point is that nilpotent groups admit a finite commutator structure that induces a corresponding finite hierarchy of irreducible representations. We will recall only those properties of nilpotent groups that are essential for understanding the derivation of the protocols; a more detailed discussion on nilpotent groups, together with explicit protocols, is provided in the subsequent sections.

Let $G$ be a nilpotent group, and let $\Irr$ denote its set of irreps. The irreps of nilpotent groups have the following properties (see Sec.~\ref{sec:NilpotentGroups} and Ref.~\cite{Gelaki2008_NilpotentFusionCategories} for more details):
\begin{enumerate}
    \item {\it Hierarchy of irreducible representations:} We can define (see Eq. \eqref{eq:Irra}) a finite hierarchy as a sequence of subsets
    \begin{align}\label{eq:Hierarchy}
        \Irr_{1} \subseteq \Irr_{2} \subseteq  \dots \subseteq \Irr_M\equiv \Irr,
    \end{align}
    where $\Irr_{1}$ is the set of all 1D irreps of $G$, and $M$ is the nilpotency class.
    \item {\it Equivalence classes on $\Irr_m$:} On each set $\Irr_m$, we can define equivalence classes as follows: For $\alpha,\beta \in \Irr_m$ we define $\alpha \sim_m \beta $ if there exists a $\gamma \in \Irr_{m-1}$ such that the decomposition of $U^\gamma_g \otimes U^\beta_g$ into irreps contains $U^\alpha_g$ (see Appendix \ref{App:IrrepGroups} for details). In this case, we say that $\alpha$ is contained (or appears in) $\gamma \otimes \beta$, or simply write $\alpha \in \gamma \otimes \beta$. 
    \item {\it Group of equivalence classes on level $m$:} For each $m$ the equivalences classes $\Irr_m/{\sim}=\{[\alpha]_m\}$ form a group with a multiplication $\circ$ defined in Sec.~\ref{sec:NilpotentGroups}. This group is isomorphic to a finite abelian group and has the following properties: (i) The identity element $[1]_m\in \Irr_m/{\sim}$ is composed of all the irreps in $\Irr_{m-1}$. (ii) The inverse element of $[\alpha]_m$ is $[\alpha^*]_m$, where $\alpha^*$ is the conjugate representation of $\alpha$ (see Appendix \ref{App:IrrepGroups}). That is, $[\alpha]_m \circ [\alpha^*]_m=[1]_m=\Irr_{m-1}$. (iii) For any $\alpha,\beta \in \Irr_m$ we have that $U^\alpha_g \otimes U^\beta_g$ decomposes only into irreps belonging to $[\alpha]_m\circ[\beta]_m\in \Irr_m/{\sim}$. That is, 
    \begin{equation}
        U^\alpha_g \otimes U^\beta_g=\bigoplus_{\gamma\in \Gamma} U_g^\gamma, \quad\text{with}\;\Gamma\subseteq[\alpha]_m\circ[\beta]_m.
        \label{eq:TensorDecompOfIrreps}
    \end{equation}
\end{enumerate}
To connect this structure to measurements in our protocols, consider a block of $n$ sites with symmetry action $U_g^{\otimes n}$. Let $P^{(n)}_\alpha$ denote the projector onto the subspace of $\mathcal{H}^{\otimes n}$ that contains \emph{all} invariant subspaces of $U_g^{\otimes n}$ that transform under the irrep $\alpha$. In the following, we will call this the subspace corresponding to an irrep $\alpha$ and ignore multiplicities. Moreover, grouping all such projectors over an equivalence class yields
\begin{equation}
    P^{(n)}_{[\alpha]_m}\equiv\sum_{\alpha\in[\alpha]_m} P_\alpha^{(n)}.
\end{equation}
By the property in Eq.~\eqref{eq:TensorDecompOfIrreps}, these projectors satisfy a fusion rule
\begin{equation}\label{eq:projectorRule}
    P_{[\alpha]_m}^{(n)}\otimes P_{[\beta]_m}^{(n')} \subseteq P_{[\alpha]_m \circ [\beta]_m}^{(n + n')},
\end{equation}
where $\subseteq$ indicates that the support of the left-hand side is contained within that of the right-hand side. This property will be central for all our protocols.

To see this, consider the case $m=M$ and suppose we first measure single sites, i.e., $n=1$. The set $\{P_{[\alpha]_M}^{(1)}\}$ forms a complete projective measurement, with outcomes labeled by group elements $[\alpha]_M\in\Irr_M/{\sim}$. Now suppose that on two neighboring sites we obtain outcomes $[\alpha]_M$ and $[\beta]_M$ with $[\alpha]_M \circ [\beta]_M = [1]_M$. In this case, the projector
\begin{equation}
P_{[\alpha]_M}^{(1)}\otimes P_{[\beta]_M}^{(1)}
\end{equation}
decomposes via Eq.~\eqref{eq:projectorRule} into projectors onto subspaces corresponding exclusively to irreps within $[1]_M=\Irr_{M-1}$. Hence, measuring the supersite (in this case composed of two sites), will only lead to outcomes corresponding to group elements in $\Irr_{M-1}/{\sim}$; the state has no support on any other subspace.

As an illustration, consider again the case of $D_8$. Here $P_{[5]_2}^{(1)}\otimes P_{[5]_2}^{(1)}$ decomposes into $1D$ projectors, i.e., 
\begin{equation}
P_{[5]_2}^{(1)}\otimes P_{[5]_2}^{(1)}\subseteq P^{(2)}_{[1]_2},
\end{equation}
as $5\otimes 5\simeq 1 \oplus 2 \oplus 3 \oplus 4$ and $[1]_2=\{1,2,3,4\}$. Thus, after blocking two outcomes corresponding to the two-dimensional irrep $5$, the supersite has support only on one-dimensional irreps.

In general, we will begin with the level-$M$ measurement $\{P_{[\alpha]_M}^{(1)}\}$, which is complete since $\Irr_{M}$ contains all irreps. We then perform a subsequent measurement with $m=M-1$ on supersites formed in the first round. Although the set $\{P^{(n)}_{[\alpha]_{M-1}}\}$ must be supplemented by an additional orthogonal projector to form a complete measurement, the state produced in the first round has no support on the corresponding subspace. Hence this additional outcome can never occur, and all measurement results lie in $\Irr_{M-1}/{\sim}$.
Suppose then that all measurements in the first round can be grouped so that their product equals the identity element $[1]_M$. Then the above argument applies to each such supersite. Therefore, in the second round we obtain exclusively measurement outcomes in $\Irr_{M-1}/\sim$. 

Repeating this procedure iteratively, and using that the hierarchy in Eq.~\eqref{eq:Hierarchy} has exactly $M$ levels, we find that after $M-1$ rounds the supersites have support only on irreps in $\Irr_1$. At this point, we can perform a measurement composed of rank-1 projectors which deterministically produces a product state in the trivial phase. In Fig.~\ref{fig:SPTProtocol}, we give an example of such a protocol using the non-abelian group $D_{16}$, which has nilpotency class three ($M=3$). The only remaining part is to show that such a protocol uses only $O(\polylog N)$ circuit depth.

To see that this procedure can be carried out consistently at every level of the hierarchy, we first note that the overall measurement outcomes in each round necessarily multiply to the identity element $[1]_m$. This follows from the fact that we measure a pure symmetric state (the initial state and the state after the $m$--th round). Such a state lies entirely in the subspace corresponding to a 1D irrep $\alpha\in\Irr_1$, and since $\Irr_1\subseteq [1]_m$ for all $m\in\{2,\dots,M\}$, the product of all measurement outcomes must equal the identity: $\circ_i[\alpha_i]_m=[1]_m$. 

This global constraint ensures that the measurement outcomes can be partitioned into subsets whose product is the identity. It then remains to show that those partitions are sized $O(\log n)$ with high probability, and that the corresponding reordering of sites can be implemented by a circuit of depth $O(\polylog N)$. As we will show in the next section, one can construct transformations in such a way that the probability of encountering cases requiring permutation circuits of depth exceeding $O(\polylog N)$ vanishes in the thermodynamic limit $N \to \infty$.

A protocol of the above form yields an asymptotically deterministic trivialization of all SPT phases: after $M-1$ measurement rounds, the state is supported only on one-dimensional irreps, allowing a final symmetric rank-1 measurement that produces a separable state in the trivial phase. This also implies the reverse transformation, as two copies of an SPT phase belong to the trivial phase~\cite{Schuch2011_Phases2}. Moreover, we will show that similar---though slightly more involved---protocols also lead to a trivialization of the non-normal (GHZ-type) phases.

\section{Finite nilpotent groups and their representations}\label{sec:NilpotentGroups}
As a preparation for the following discussion, we need to recall some known facts about nilpotent groups and their (irreducible) representations~\cite{Isaacs2006_CharacterTheoryFiniteGroups,DummitFoot2003_AbstractAlgebraTextbook}. Let $G$ be a finite group. The group commutator is defined by $[g,h]=g^{-1}h^{-1}gh$ for all $g,h\in G$. The commutator subgroup, denoted by $[G,G]$, is the subgroup generated by all elements of the form $[g,h]$, with $g,h\in G$. The \emph{lower central series},
\begin{equation}
    G\equiv G_0\unrhd G_1 \unrhd G_2 \unrhd \cdots,
\end{equation}
where $\unrhd$ indicates a normal subgroup, is a descending series of subgroups that is obtained by repeatedly taking commutator subgroups of $G$, i.e., $G_0\equiv G$ and $G_{m+1}=[G_m,G]$ for $m\geq 0$. The crucial property of nilpotent groups is that there exists an $M\in\mathbb{N}$ at which the lower central series terminates in the trivial group. That is,
\begin{equation}
    G=G_0\unrhd G_1 \unrhd \dots \unrhd G_M=\{e\}.
\end{equation}
A nilpotent group is \emph{class-$M$ nilpotent} if $M$ is the smallest such integer. Abelian groups are class-one nilpotent.

Using the lower central series of a class-$M$ nilpotent group $G$, we define the following sequence of subsets of the set of irreducible representations of $G$:
\begin{equation}\label{eq:HierachyIrr}
    \{1\}\equiv\text{Irr}_0 \subseteq \text{Irr}_1 \subseteq \dots \subseteq \text{Irr}_M\equiv\Irr,
\end{equation}
where $\Irr_m$ is the set of all irreps of $G$ for which the representation matrix of any element in $G_m$ equals the identity, i.e.,
\begin{equation}
    \text{Irr}_m \equiv \{ \alpha \in \text{Irr} : U_{g}^\alpha = \id_{d^\alpha}, \forall g\in G_m  \}
    \label{eq:Irra},
\end{equation}
with $m\in\{0,\dots,M\}$ and $d^\alpha$ the dimension of the irrep. For instance, we have that $\Irr_0$ contains only the trivial representation, $\Irr_1$ contains all 1D irreps of $G$, and $\Irr_M$ contains all irreps of $G$. Note that by Eq.~\eqref{eq:Irra}, if $\alpha,\beta\in\Irr_m$ then both are trivial on $G_m$, and hence the tensor product $\alpha\otimes\beta$ is also trivial on $G_m$. Therefore every irrep appearing in the decomposition of $\alpha\otimes\beta$ also lies in $\Irr_m$.

We now define an equivalence relation on each of the sets $\Irr_m$. It is based on the properties of tensor product of irreps of nilpotent groups (see also Ref.~\cite{Gelaki2008_NilpotentFusionCategories}). Two irreps $\alpha, \beta\in\Irr_m$ are said to be equivalent if there exists an irrep $\gamma\in \Irr_{m-1}$ (one step below in the hierarchy), such that $\alpha$ appears in the decomposition of $\gamma\otimes\beta$. That is,
\begin{equation}
    \beta \sim_m \alpha \text{ if } \exists \gamma\in \Irr_{m-1} \text{ : } \alpha\in \gamma\otimes \beta, \label{eq:EquivalenceRelation@Levelm}
\end{equation}
where $\alpha\in\gamma\otimes\beta$ means that $U^\alpha_g$ occurs as a summand in the decomposition of the representation $U^\gamma_g\otimes U^\beta_g$. It is verified in Appendix~\ref{App:IrrepGroups} that this is indeed an equivalence relation. We denote the equivalence class of $\alpha$ under $\sim_m$ by $[\alpha]_m$, and we drop the subscript $m$ whenever it is clear from context.

Let us also introduce a group multiplication, $\circ$, on $\Irr_m/{\sim}$ by defining $[\alpha]\circ[\beta]$ as those irreps that appear in the decomposition into irreps of $\alpha'\otimes\beta'$  for all $\alpha'\sim\alpha$ and $\beta'\sim\beta$. That is, 
\begin{equation}\label{eq:group_mult}
    [\alpha] \circ [\beta] = \{\delta \in \Irr_m : \exists \alpha'\in[\alpha], \beta'\in[\beta] : \delta\in \alpha'\otimes \beta' \}.
\end{equation}
In Appendix~\ref{App:IrrepGroups}, we show that $(\Irr_m/{\sim},\circ)$ does indeed form a group, and moreover, that this group is isomorphic to the abelian quotient group 
\begin{equation}\label{eq:groupIso}
    G_{m-1}/G_m \simeq \Irr_m/{\sim}.
\end{equation}
Furthermore, in Appendix~\ref{app:Example}, we provide an example of a nilpotent group and explain its properties regarding the structures discussed above.

The identity element of $\Irr_m/{\sim}$ is the class $[1]_m$, consisting precisely of the irreps at the level below in the hierarchy; i.e., $[1]_m=\Irr_{m-1}$. This follows directly from the definition of $\sim_m$, since tensoring with any irrep in $\Irr_{m-1}$ leaves the equivalence class unchanged. Moreover, for any $\alpha \in \Irr_m$, its conjugate representation $\alpha^*$ belongs to the group inverse in $(\Irr_m/\sim,\circ)$, i.e., $\alpha^*\in[\alpha]^{-1}$. Indeed, the tensor product of $\alpha$ and $\alpha^*$ decomposes into irreps belonging to $\Irr_{m-1}$ (see Appendix~\ref{App:IrrepGroups}).

The group structure on the equivalence classes defined in Eq.~\eqref{eq:group_mult} 
naturally appears in the tensor product of two irreps at level~$m$. Specifically, one finds that
\begin{equation}\label{eq:TensorProductsDecomposesIntoOneEquivClass}
    U^\alpha_g \otimes U^\beta_g=\bigoplus_{\gamma\in \Gamma\subseteq[\alpha]\circ[\beta]} U_g^\gamma,
\end{equation}
so every irrep $\gamma$ appearing in the decomposition of $\alpha \otimes \beta$ lies in the equivalence class $[\alpha] \circ [\beta]$ (although not all such irreps necessarily appear). This property plays a central role in our protocols, as it ensures that combining measurement outcomes produces supersites with support solely within a single equivalence class, enabling the hierarchical descent used in the subsequent sections.

To this end, we will repeatedly use projective measurements whose projectors act on subspaces corresponding to the equivalence classes in $\Irr_m$. We write $\{P_{[\alpha]}^{(n)}\}_{[\alpha]\in \Irr_m}$ for the projection onto the subspaces of $U_g^{\otimes n}$ corresponding to irreps in $[\alpha]$; that is (see Ref.~\cite{Serre1977}, Theorem 8),
\begin{equation}\label{eq:ProjectorOntoIrrepClasses}
    P_{[\alpha]}^{(n)}=\sum_{\alpha\in [\alpha]}P_\alpha^{(n)}=  \sum_{\alpha\in [\alpha]} \sum_{g\in G} \frac{d^{\alpha} \bar{\chi}^{\alpha}_g}{|G|} U_g^{\otimes n}, 
\end{equation}
with $P_\alpha^{(n)}$ the projector onto the subspace corresponding to $\alpha\in\Irr$ in $U_g^{\otimes n}$. Note that due to Eq.~\eqref{eq:TensorProductsDecomposesIntoOneEquivClass}, we have that
\begin{equation}
    P_{[\alpha]}^{(n)}\otimes P_{[\beta]}^{(n')} \subseteq P_{[\alpha] \circ [\beta]}^{(n + n')},
\end{equation}
where $\subseteq$ indicates that the support on the left-hand side is contained in the support of the right-hand side. This identity will be used repeatedly throughout this work. In particular, it shows that if we measure $n$ sites and obtain outcomes $[\alpha_i]\in \Irr_m/{\sim}$ whose product satisfies
\begin{equation}
\circ_i [\alpha_i]=[1]\in\Irr_m/{\sim} \ = \Irr_{m-1},
\end{equation}
then the combined $n$ sites can be treated as a single supersite with their support entirely contained in $\Irr_{m-1}$.

We may then proceed by measuring $\{P_{[\beta]}^{(n)}\}_{[\beta]\in \Irr_{m-1}/{\sim}}$. Although this measurement must be supplemented by an additional orthogonal projector to be complete, the state obtained from the preceding measurement has no support on the corresponding subspace, so that outcome can never occur. As a result, in each round $m$ we deterministically project the state onto subspaces corresponding to irreps in $\Irr_{m-1}$, descending one level in the hierarchy at each step. Iterating this argument, we eventually reach measurements supported on $\Irr_1$, the level of one-dimensional irreps (the abelian subspaces), which is the final stage of the protocol.

In the following, we will evaluate equations that involve Eq.~\eqref{eq:ProjectorOntoIrrepClasses}. The following technical lemma simplifies those calculations by providing a more compact expression for the above projectors. The proof of this lemma can be found in Appendix~\ref{app:SumProof}.
\begin{lem}\label{lem:IrrepSumTrick}
Let $\Irr_m/{\sim}\cong G_{m-1}/G_m=\{[\alpha]\}$. Then, we have the following:
\begin{equation}
P^{(n)}_{[\alpha]}= \sum_{g\in G_{m-1}} \frac{e^{-i\varphi^{[\alpha]}([g])}}{\abs{G_{m-1}}} U_g^{\otimes n},
\end{equation}
where $[g]$ denotes the canonical projection map from $G_{m-1}$ to $G_{m-1}/G_m$, and  $e^{-i\varphi^{[\alpha]}}\in \Irr( G_{m-1}/G_m)$ is a 1D irrep of $G_{m-1}/G_m$ defined in Appendix~\ref{app:SumProof}.
\end{lem} 

\section{Trivialization of symmetry-protected phases}\label{sec:ProtocolSPT}
In this section, we show that all SPT phases of MPS with respect to a finite nilpotent group $G$ trivialize under $G$-CMF transformations. We describe a protocol that transforms an SPT state into a product state in the trivial phase while respecting the symmetry described by a finite nilpotent group. The protocol proceeds in $M$ rounds of measurements interleaved with low-depth circuits. The reverse transformation then follows immediately as two copies of an SPT state belong to the trivial phase~\cite{Schuch2011_Phases2}; transforming one copy to the trivial phase leaves an SPT state on the second.

We begin by choosing a representative state for the SPT phase. We then explain the details of our protocol and demonstrate that it successfully achieves the desired transformation. Then, we analyze the resources required for its implementation, including its circuit depth and probability of success.

\begin{figure*}
    \centering
    \includegraphics[width=1\linewidth]{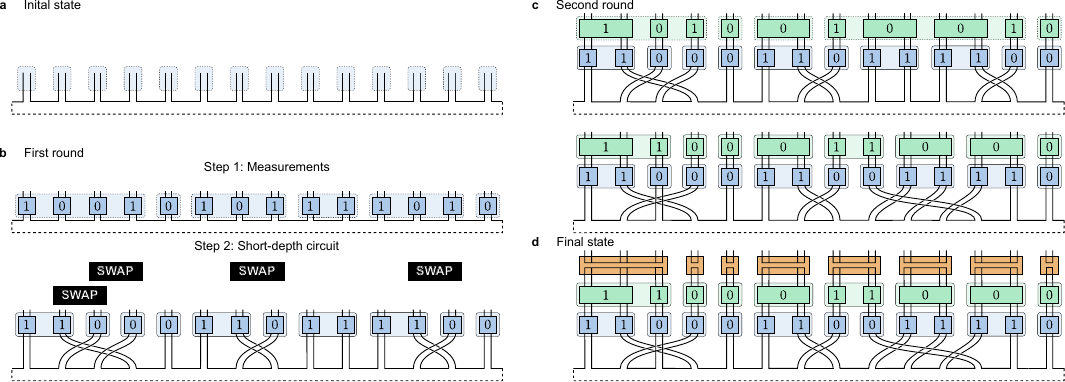}
    \caption{An example of the SPT protocol for the non-abelian class-3 nilpotent group $D_{16}$, transforming an SPT state to the trivial phase.
    (a) We depict the initial non-trivial state, $\ket{\rm SPT}$, corresponding to a chain of Bell pairs between nearest neighbors. This state is symmetric under $U_g^{\otimes N}$, with the effective symmetry defined in Eq.~\eqref{eq:Ug}.
    (b) In the first round of the protocol, we project onto the equivalence classes in $\Irr_3/{\sim}$, with projectors represented by blue boxes. Due to the isomorphism in Eq.~\eqref{eq:groupIso} we have that $\Irr_3/{\sim} \cong G_2/G_3\cong \mathbb{Z}_2$. Therefore, measurement outcomes can be labeled by elements of $\mathbb{Z}_2$. Note, that the total product of outcomes is equal to $0$, the identity element. Next, we partition the output into minimal and complete connected substrings. Here, the longest minimal connected substring that multiplies to the identity is of length $4$. Therefore, we need $L^{(1)}\ge 16$ to ensure the protocol does not fail in this step. In a second step we use a SWAP circuit to permute these minimal connected substrings to get complete substrings of length at most $|G_2|/|G_3|=2$. This is to ensure that subsequent measurements act on only $O(1)$ sites. By inspection, we see that the second round will measure $N^{(2)}=10$ supersites. 
    (c) In the second round of the protocol, we then project onto the equivalence classes in $\Irr_{2}/{\sim}$, with projectors represented by green boxes. Outcomes are labeled by $\Irr_{2}/{\sim}\cong G_1/G_2 \cong \mathds{Z}_2$. As before, we use a circuit to permute the outcome. Then, substrings have support on subspace of 1D irreps. Finally, in (d), we project onto 1D subspaces, represented by orange boxes. As these measurements are rank-1, the resulting state is a symmetric pure product state with local sites of size $O(|G|)$.
    }
    \label{fig:SPTProtocol}
\end{figure*}

\subsection{Representative for the phase}
Let $G$ be a class-$M$ nilpotent group with a nontrivial second cohomology group $H^2(G,U(1))$. Associated with each $\mu\in H^2(G,U(1))$ is a cocycle $\gamma^\mu:G\times G\rightarrow U(1)$, and the so-called \emph{projective $\mu$-regular representation} is given by~\cite{karpilovsky1994_ProjTraceCondition}
\begin{equation}
    \omega^\mu_g=\sum_{h\in G} \gamma^\mu_{g,h}\ketbra{g\circ h}{h}.
\label{eq:ProjectiveRegularRep}
\end{equation}
As before, $\{\ket{h}\}_{h\in G}$ denotes some orthonormal basis and $\circ$ the group multiplication in $G$. This generalizes the linear regular representation of a group $G$, and if $\mu$ is trivial the projective $\mu$-regular representation reduces to the linear regular representation of $G$. This representation is reminiscent of the generalized Pauli matrices, sometimes also referred to as clock- and shift matrices. Indeed, they fulfill similar properties as the generalized Paulis. For instance, they are mutually orthogonal under the Hilbert-Schmidt inner product, i.e.,
\begin{equation}\label{eq:ProjectiveIrrepsOrthogonal}
    \mathrm{tr}(\omega^{\mu\dagger}_g\omega^\mu_h)=\abs{G}\delta_{g,h},
\end{equation}
for all $g,h\in G$, and, as projective representations, they are closed under multiplication up to phases, i.e., $\omega^\mu_g\omega^\mu_h=\gamma_{gh}^\mu\omega^\mu_{gh}$.

As we consider the on-site symmetry to be the regular representation, for any unitary representation there is a finite blocking length $k$ at which that representation appears a subrepresentation of $(U_g^{\rm reg})^{\otimes k}$ (see Section \ref{sec:prelim}). This allows us to freely choose a representative state for the phase as follows: Fix a non-trivial $\mu$ and define the effective on-site symmetry
\begin{equation}\label{eq:Ug}
    U_g=(\omega_g^{\mu})^*\otimes \omega_g^\mu.
\end{equation}
Then, a representative for the SPT phase can be chosen as~\cite{Chen2011_Phases1,Schuch2011_Phases2}
\begin{equation}\label{eq:sptstate}
    \ket{\rm SPT} = \bigotimes_{i=1}^N \ket*{\Phi^+_{\abs{G}}}_{i,i+1},
\end{equation}
which is a chain of maximally-entangled Bell states between nearest neighbors, see Fig.~\ref{fig:SPTProtocol}(a). This will be our representative state for the SPT phases. Next, we will show how to transform this state to a product state in the trivial phase.

\subsection{Protocol}
Having fixed representatives of the phases, we now present a protocol that maps the state $\ket{\rm SPT}$ in Eq.~\eqref{eq:sptstate} to a separable state in the trivial phase. We will explain the first round in more detail, and briefly comment on subsequent rounds, which then proceed similarly. In Fig.~\ref{fig:SPTProtocol} we provide an example illustrating the following protocol.

\subsubsection{First round}
On each site, we perform the projective measurement $\{P_{[\alpha]}^{(1)}\}$, where $P_{[\alpha]}^{(1)}$, as defined in Eq.~\eqref{eq:ProjectorOntoIrrepClasses}, projects onto the subspace corresponding to those irreps in $U_g$ that belong to $[\alpha]\in \Irr_M/{\sim}$. We can label the measurement outcome by a vector
\begin{equation}
    \vec{x}^{(1)} = (x^{(1)}_{1},x^{(1)}_{2},\dots, x^{(1)}_{N}),
\end{equation}
where $x^{(1)}_i\in \Irr_M/{\sim}$, and $N\equiv N^{(1)}$ is the number of sites with the upper index referring to the measurement round.

On the level of the entire system of $N$ sites, the outcomes are correlated. The only outcomes to occur are those where $x_1^{(1)}\circ\dots\circ x_N^{(1)}=[1]\in \Irr_{M}/\sim$. As mentioned before, the reason for that is that the global state is pure and symmetric and thus must belong to the abelian subspace, i.e., the subspace of $U_g^{\otimes N}$ corresponding to irreps in $\Irr_1$.
As $\Irr_1\subseteq [1]\in \Irr_{M}/\sim$ we have that~\cite{note4}
\begin{multline}
    x_1^{(1)}\circ\dots\circ x_N^{(1)}\neq [1] \\
    \Rightarrow \bra{\rm SPT} P_{x^{(1)}_1}\otimes\dots\otimes P_{x^{(1)}_N}\ket{\rm SPT}=0. \label{eq:FirstSPTParityConstraint}
\end{multline}
After the first round of measurements, we can partition the outcome vector $\vec{x}^{(1)}$ into the shortest connected substrings that individually multiply to the identity element, see Fig.~\ref{fig:SPTProtocol}(b) for an example. In what follows, we refer to such substrings—whose product equals the identity—as 
$m$-complete substrings, or simply complete when $m$ is clear from the context.

If the longest of such substrings is greater than $L^{(1)}$, where $L^{(1)}$ is some yet to be specified function of $N^{(1)}=N$, the protocol terminates and fails. Later we will show that if we choose $L^{(1)}=O(\log N)$ with a sufficiently large constant prefactor, then the probability that there is a substring longer than $L^{(1)}$ vanishes in the thermodynamic limit.

Assuming that the longest of such substrings is not greater than $L^{(1)}$, we use a SWAP circuit to permute each substring in such a way that one obtains \emph{minimal} connected substrings that are complete and are of length at most $\abs{G_{M-1}}/\abs{G_M}$ (see Fig.~\ref{fig:SPTProtocol}(c) for an example). This follows from the fact that $\abs{G_{M-1}}/\abs{G_M}$ is finite. By doing so, we have transformed our vector of outcomes $\vec{x}^{(1)}\mapsto \tilde{\vec{x}}^{(1)}$, where the latter can be partitioned as 
\begin{align}
    \tilde{\vec{x}}^{(1)}=(\tilde{\vec{x}}_i^{(1)})_{i=1}^{N^{(2)}},
\end{align}
where $N^{(2)}=N^{(2)}(\vec{x}^{(1)})$ is the number of such substrings, which depends on $\vec{x}^{(1)}$.
Since all substrings $\tilde{\vec{x}}^{(1)}_i$ are $M$-complete, the corresponding sites are supported only on subspaces corresponding to irreps in $[1]_M=\Irr_{M-1}$.
Consequently, in the next round we will project onto subspaces corresponding to irreps in $\Irr_{M-1}$.

\subsubsection{Second round}
In the second round, we analyze the system at the level of the substrings defined in the first round. For each site corresponding to a substring, we perform the measurement $\{P^{(N_{i}^{(1)})}_{[\beta]}\}$ where $N_{i}^{(1)}\le|G_{m-1}|/|G_m|$ denotes the length of the $i^{th}$ substring from the first round, and $[\beta]\in\Irr_{M-1}/\sim$. Sinve the projective measurements across different rounds commute, we again obtain a parity constraint, $[\beta_1]\circ \dots \circ [\beta_{N^{(2)}}] =[1]\in\Irr_{M-1}$. We then proceed in the same way as in the first round.

\subsubsection{Subsequent rounds}
After performing $M-1$ rounds of measurements and the corresponding circuits, all sites are locally supported on irreps in $\Irr_{1}/{\sim}=\Irr_1$ which contains all 1D irreps. Consequently, the state is locally supported on a subspace that decomposes entirely into 1D irreps. This allows us, in a final measurement, to disentangle the state and apply local corrections to obtain a translation-invariant state in the trivial phase.

\subsection{Resources} \label{Sec:SPTResources}
Now let us examine the resources required for the protocol. A round begins with a measurement and ends after any subsequent circuits. In round $m\in\{1,\dots ,M\}$, we project onto the different subspaces corresponding to equivalence classes of irreps in $\Irr_{M-(m-1)}/{\sim} \ \cong G_{M-m}/G_{M-(m-1)}$. After permuting outcomes and forming minimal connected substrings that are complete, each supersite consists of at most $\abs{G_{M-m}}/\abs{G_{M-(m-1)}}$ supersites from the previous round. Thus, after $k$ rounds, the supersites consist of at most $\prod_{m=1}^k\abs{G_{M-m}}/\abs{G_{M-(m-1)}}=\abs{G_{M-k}}/\abs{G_{M}}=\abs{G_{M-k}}$ original sites, where the last equality follows from $G_M=\{e\}$. Therefore, the size of the measured sites is independent of the system-size. Moreover, as the number of rounds is finite, at the end of the protocol each supersite is of size $O(\abs{G})$.

Next, let us consider the cumulative circuit depth of the protocol.
In round $m$, we perform permutations on substrings of length at most $L^{(m)}$. Recall, that if the longest minimal connected substring that is complete and has length greater than $ L^{(m)}$, the protocol terminates and fails. An arbitrary permutation of $n$ sites can be performed using a SWAP circuit of depth at most $4n$~\cite{Bereg2016_MinSwapCircuit}. Thus, the SWAP circuit in round $m$ has depth at most $O(L^{(m)}(N^{(m)}))$ at the length-scale of round $m$. We note again, that each round effectively rescales the length-scale at which we operate. Moreover, $N^{(m)}\le N$. Therefore, we may upper bound the cumulative circuit depth of the protocol with $\max_m O(|G_m| L^{m}(N))$. In the subsequent sections, we will show that choosing $L^{(m)}(n)=O(\log n)$, with a sufficiently large constant prefactor, is sufficient to ensure the probability of success converges to $1$ as $N\rightarrow \infty$. For such a choice, the cumulative circuit depth is $O(\log N)$.

\subsection{Probability distribution of measurement outcomes}  \label{sec:SPTProbabilities}

In this section, we show that the outcome distribution in each round is uniform, subject to a parity constraint. It is sufficient to analyze the second round, as the arguments for all other rounds are analogous.

The probability distribution of the outcome of the second round, $\vec{x}^{(2)}$, depends on the outcome in the first round, $\vec{x}^{(1)}$. Indeed, the number of supersites measured in a given round depends on the measurement outcome of the preceding round. Recall that, given an outcome $\vec{x}^{(1)}$ in the first round, we permute it to obtain $\tilde{\vec{x}}^{(1)}$. The string $\tilde{\vec{x}}^{(1)}$ then decomposes into $N^{(2)}$ complete substrings of length at most $|G_{M-1}|/|G_M|$.

For simplicity, consider the case where no permutation is required (cases involving permutations between measurements will be addressed at the end). In this case, we can partition $\vec{x}^{(2)}$ and $\vec{x}^{(1)}$ using the same indices. Thus, we can write
\begin{align}\label{eq:NoCircuitsNeeded}
\vec{x}^{(2)}&=(\vec{x}^{(2)}_i)_{i=1}^{N^{(2)}}, \nonumber \\
\vec{x}^{(1)}&=(\vec{x}^{(1)}_i)_{i=1}^{N^{(2)}}=((x^{(1)}_{ij})_{j=1}^{N_i^{(1)}})_{i=1}^{N^{(2)}}.
\end{align}
Note that, since the outcome in the first round satisfied a parity constraint (see Eq.~\eqref{eq:FirstSPTParityConstraint}) and required no rearrangement, 
\begin{equation}\label{eq:condition}
    \prod_{j} x_{ij}^{(1)} = [1],\quad\forall i.
\end{equation}
Then, we can express the conditional probability of observing $\vec{x}^{(2)}$ in the second round, given that we have observed $\vec{x}^{(1)}$ in the first round, as 
\begin{align}
p(\vec{x}^2 | \vec{x}^1) &=
\frac{ \tr \left( \begin{aligned}
\includegraphics[width=0.45\linewidth]{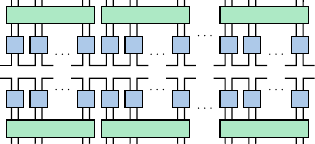}
\end{aligned} \right) }{\tr \left(\begin{aligned}
\includegraphics[width=0.45\linewidth]{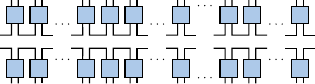}
\end{aligned} \right)}\nonumber\\
&=\quad \  \frac{
\begin{aligned}
\includegraphics[width=0.45\linewidth]{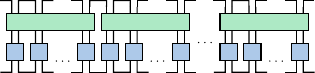}
\end{aligned}
}
{
\begin{aligned}
\includegraphics[width=0.45\linewidth]{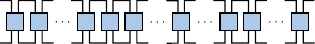}
\end{aligned}} 
\label{eq:SPTEq1}
\end{align}
where the blue boxes correspond to projectors onto subspaces corresponding to equivalence classes in $\Irr_{M-1}/\sim$ and the green boxes correspond to projectors onto subspaces corresponding to equivalence classes $\Irr_{M-2}$ (see Fig.~\ref{fig:SPTProtocol} for comparison). The second equality follows from the cyclic property of the trace and the fact that 
\begin{equation}
    [P_{x^{(2)}_i},\bigotimes_j P_{x_{ij}^{(1)}}]=0, \quad\forall i
\end{equation}
due to Eq.~\eqref{eq:condition}. To evaluate the conditional probability, consider the first term in the numerator
\begin{align}
  & \begin{aligned}
         \includegraphics[width=0.16\linewidth]{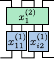}
 \end{aligned}
 \label{eq:SPTLoop1}
\end{align}
Using Lemma \ref{lem:IrrepSumTrick} and the fact that $U_g=\omega_g^*\otimes\omega_g$, this evaluates to
\begin{align}
 \sum_{\substack{b_1\in G_{M-2}\\ a_{11},a_{12} \in G_{M-1}}} 
  & \frac{e^{-i\phi^{x^{(2)}_i}([b_1])}}{|G_{M-2}|} \frac{e^{-i\phi^{\tilde{x}^{(1)}_{ij}}([a_{i1}])}}{|G_{M-1}|}
 \frac{e^{-i\phi^{\tilde{x}^{(1)}_{ij}}([a_{i1}])}}{|G_{M-1}|}\nonumber\\
 & \times\ \  \begin{aligned}
    \includegraphics[width=0.4\linewidth]{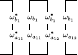}
 \end{aligned}
\end{align}
Consider the loop in this equation. Each factor $\omega_{b_i}\omega_{a_{ij}}$ appears together with its conjugate in $U_g$, so the phases arising from the multiplication of projective representations cancel. Moreover, the projective $\mu$-regular representations are orthogonal by Eq.~\eqref{eq:ProjectiveIrrepsOrthogonal}. Consequently, these loops yield delta functions of the form
\begin{align}
|G| \delta\left(b_{i}a_{ij}= b_{(ij\oplus1)_1}a_{ij\oplus 1}\right),
 \label{eq:SPTEq2}
\end{align}
where $(ij\oplus1)$ denotes the next index in the chain.

Now we consider that for each substring of $\vec{x}^{(1)}$, we have
\begin{align}
&b_ia_{ij}=b_ia_{ij'}\,\Rightarrow\, a_{ij}=a_{i1}\,\Rightarrow \,[a_{ij}]=[a_{i1}],
\end{align}
for all $i,j$, where $[\cdot]$ denotes the canonical projection from $G_{M-1}\rightarrow G_{M-1}/G_{M}$.
Moreover, if we consider the boundaries between substrings, we get that for all $i$
\begin{align}
b_ia_{iN^{(1)}_{i}}=b_{i\oplus1} a_{i\oplus1,1}\,\Rightarrow\, [b_{i}]=[b_{1}],
\end{align}
where $[\cdot]$ denotes the canonical projection from $G_{M-2}\rightarrow G_{M-2}/G_{M-1}$.
Thus, the numerator of Eq.~\eqref{eq:SPTEq1} becomes
\begin{align}
|G|^N \sum_{\substack{b_1\in G_{M-2}\\ a_{i1} \in G_{M-1}}} \Bigg( \prod_{i=1}^{N^{(2)}} \frac{e^{-i\phi^{x^{(2)}_i}([b_1])}}{|G_{M-2}|} \prod_{j=1}^{N^{(1)}_i} \frac{e^{-i\phi^{\tilde{x}^{(1)}_{ij}}([a_{i1}])}}{|G_{M-1}|}\Bigg). 
\end{align}
Using the fact that the phases are 1D irreps, and Eq.~\eqref{eq:condition} this evaluates to 
\begin{align}
|G|^N \left( \frac{|G_M|}{|G_{M-1}|}\right)^{N-1} \left( \frac{|G_{M-1}|}{|G_{M-2}|}\right)^{N^{(2)}-1} \delta\left(\prod_i x^{(2)}_i = [1] \right).
\end{align}
It is easily seen through the same arguments that the denominator of Eq.~\eqref{eq:SPTEq1} evaluates to the first term in the equation above. Thus
\begin{align}
p(\vec{x}^{(2)}|\vec{x}^{(1)})=\left( \frac{|G_{M-1}|}{|G_{M-2}|}\right)^{N^{(2)}-1} \delta\left(\prod_i x^{(2)}_i = [1] \right),
\end{align}
and therefore, $p(\vec{x}^{(2)}|\vec{x}^{(1)})$ is uniform with a parity constraint.

Let us consider the effect of a permutation circuit between measurement one and two. In this case, the delta function in Eq.~\eqref{eq:SPTEq2} becomes 
\begin{align}
\delta \left( b_{(ij)_1}a_{ij}= b_{(\sigma(ij\oplus 1)\oplus 1)_1}a_{\sigma(ij\oplus 1)} \right),
\end{align}
where $\sigma$ is the permutation applied to the sites. However, as $\sigma$ is a bijection, all delta functions can be rearranged to obtain the original conditions. Thus the permutations have no effect on the probability. Finally, it is clear how the above arguments extend to all other rounds. Thus, the above arguments demonstrate that the probability distribution of outputs in every round of the protocol is uniform with a parity constraint.

\subsection{Probability of success}
We now show that, when we restrict ourselves to circuits of depth $O(\log N)$, the probability of success converges to one as $N\rightarrow\infty$. Therefore, let us analyze the probability of failure in each given round. 

\subsubsection{First round} 
A necessary condition for the protocol to fail in the first round is that there exists at least one connected and complete substring $\vec{y}$ of length $L^{(1)}$ in $\vec{x}^{(1)}$, i.e.,
\begin{equation}\label{eq:L1substring}
    y_1 \circ \dots \circ y_{d'} \ne [1]\in G_{M-1}/G_M,\quad \forall d'=1,\dots,L^{(1)}.
\end{equation}
In the following, we refer to such strings as ``$L^{(1)}$ substrings". This condition is not sufficient: such a substring may appear without causing failure due to the periodic boundary condition~\cite{note5}.

As the above condition is necessary, the probability of the protocol failing in the first round is upper bounded by the probability of obtaining such a substring. Moreover, we can derive a loose upper bound on the probability of observing such a substring by summing, over all starting positions $j$, the probabilities of all strings that contain an $L^{(1)}$-substring beginning at position $j$.

Consider an $L^{(1)}$-substring beginning at position $j$. If the measurements were performed sequentially, obtaining an $L^{(1)}$-substring would correspond to each subsequent measurement not being the inverse of the product of the previous measurements (see Eq.~\eqref{eq:L1substring}). Since the probability distribution is uniform up to a global constraint, the inverse would occur in each subsequent measurement with probability $1/(|G_{M-1}|/|G_{M}|)$ (except for the last measurement, which would be fixed by the parity constraint). Therefore, the probability of an outcome containing an $L^{(1)}$-substring beginning at position $j$ is upper bounded by $\left(1-|G_M|/|G_{M-1}|\right)^{L^{(1)}}$, cf. Ref.~\cite{Schilling1990_CoinFlippingLongestRunOfHeads}. This argument applies to each starting position $j$, and thus
\begin{align}
    p_\text{fail}^{(1)} \le N \left(1-\frac{|G_{M}|}{|G_{M-1}|}\right)^{L^{(1)}}.
\end{align}
Therefore, if we choose 
\begin{equation}
    L^{(1)}=C_1 \log(N),
\end{equation}
with $C_1 = \frac{-(1+R_1)}{\log{\left(1-\frac{|G_{M}|}{|G_{M-1}|}\right)}}$ for any $R_1>0$, then 
\begin{align}
    p_\text{fail}^{(1)} \le \frac{1}{N^{R_1}}\rightarrow 0.
\end{align}

\subsubsection{Second round}
If the protocol succeeds in the first round with outcome $\vec{x}^{(1)}$, then after permuting the sites, the second-round measurement is performed on at least $N/(|G_{M-1}|/|G_M|)=N/|G_{M-1}|$ sites. The protocol then fails if the measurement outcome contains an $L^{(2)}$-substring. The failure probability can be upper bounded as in the first round. Namely, by choosing $L^{(2)}(N)=C_2 \log(N)$, with $C_2=\frac{-(1+R_2)}{\log\left(1-\frac{|G_{M-1}|}{|G_{M-2}|}\right)}$ and $R_2>0$ we have
\begin{equation}
    p_{\text{fail}|\vec{x}^{(1)}}^{(2)}\le \left(\frac{|G_{M-1}|}{N}\right)^{R_2},
\end{equation}
for all $\vec{x}^{(1)}$. Therefore, the probability of the protocol failing in the second round is also upper bounded by $\left(|G_{M-1}|/N\right)^{R_2}$, and, therefore, the probability of having failed in the first \emph{or} second round is bounded by
\begin{equation}
   p_{\text{fail}}^{(2)} \le \frac{1}{N^{R_1}} + |G_{M-1}|^{R_2} \frac{1}{N^{R_2}}  \rightarrow 0
\end{equation}
with $R_1,R_2>0$.

\subsubsection{Subsequent rounds}
The same arguments apply in all subsequent rounds. That is, if in round $m$ we allow a circuit depth of
\begin{equation}
    L^{(m)}(N) = C_m \log(N),
\end{equation}
with
\begin{equation}
    C_m = \frac{-(1+R_m)}{\log(1-\frac{|G_{M-(m-1)}|}{|G_{M-m}|})}
\end{equation}
and $R_m>0$, then, as the protocol terminates after $M$ rounds, the probability that the protocol fails converges to zero as $N\rightarrow \infty$. Combined with the resource analysis from the Section \ref{Sec:SPTResources}, this demonstrates that it is possible to convert nilpotent SPT phases to the trivial phase asymptotically deterministically with symmetry preserving CMF.

\section{Trivialization of non-normal phases} \label{sec:GHZProtocol}

In this section, we show the trivialization of the non-normal (or GHZ) phases.
To begin, the protocol used to transform the SPT state into the trivial phase also applies to the GHZ state. In this case, the probability distribution of measurement outcomes in each round remains uniform, subject to a parity constraint (see Appendix~\ref{Appendix:VerificationGHZIsOutputed}). Therefore, the same protocol applies.

Performing the reverse transformation---preparing the GHZ state from a product state---requires additional considerations. In essence, the protocol proceeds as follows (see Fig.~\ref{fig:GHZProtocolPart1} for an example, cf. Ref.~\cite{Gunn2025} for abelian symmetries): First, we locally prepare the fiducial state defined in Eq.~\eqref{eq:fiducial} of the GHZ tensor. We then wish to perform a Bell-type measurement on adjacent virtual sites to obtain a GHZ state. However, since the state is locally supported on non-abelian subspaces of $U_g^{\otimes 2}$, we first employ an approach similar to the one used in the SPT protocol to successively reduce the support to abelian subspaces. In contrast to the SPT case, the outcomes do not admit a parity constraint. Therefore, we iteratively introduce auxiliary qubits, which are measured until the required parity constraint is satisfied. Crucially, we demonstrate that $O(\log(N))$ auxiliary systems are sufficient to ensure the probability of failure converges to zero. Finally, we perform a Bell-type measurement, including appropriate local corrections, to obtain the desired GHZ state. In the following, we analyze the resource requirements, examine the probability distribution of measurement outcomes, and show that, when restricted to a cumulative circuit depth of $O(\log(N))$, the probability of success approaches one in the thermodynamic limit.

\begin{figure*}
    \centering
    \includegraphics[width=1\linewidth]{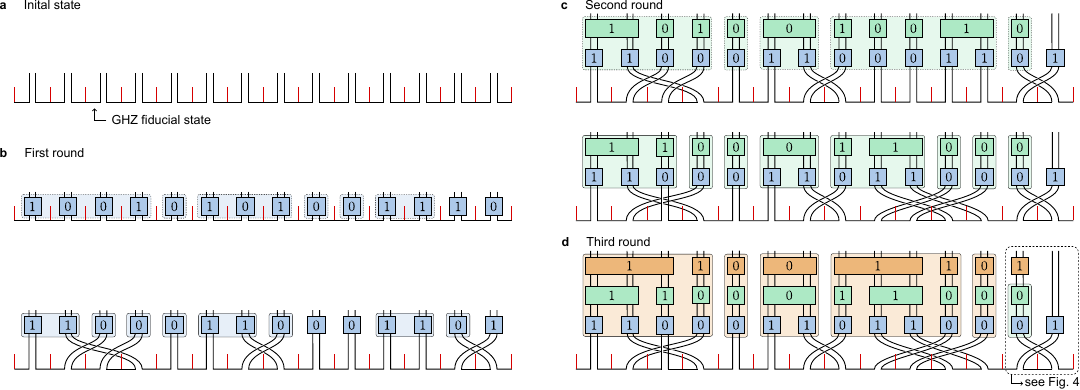}
    \caption{Example of Part 1 of the protocol transforming a state in the trivial phase to the GHZ state for the class-$4$ nilpotent group $D_{32}$---the dihedral group with thirty-two elements.
    For this group we have that the first three levels of the hierarchy are equipped with the group structure $\Irr_m\cong G_{m-1}/G_m\cong \mathbb{Z}_2$ for $m=4,3,2$, i.e., the outcomes in each round can be labeled by $0,1\in\mathbb{Z}_2$.
    The protocol begins with the state in (a)---local copies of the fiducial state of the GHZ tensor, where the middle qudits, indicated by red lines, represent physical sites that remain untouched until the end of the protocol. Clearly the state is separable between sites and therefore in the trivial phase. After using a nearest-neighbor SWAP circuit to move each right auxiliary qudit to its right neighbor, the first round begins by performing measurements similar as in the SPT case.
    Note that, unlike the SPT case, the measurement outcomes are not guaranteed to multiply to the identity.
    Therefore, in each round, we have the possibility of a ``remainder" parity. 
    For example, in (b), the product of all measurement outcomes yields $1$ instead of $0$.
    As a result, there is a remainder at the end.
    The sites are permuted so that the remainder is moved to the end of the string.
    The protocol then proceeds to the next round, (c), leaving the remainder unmeasured until Part 2 of the protocol.
    Although measurement outcomes are not guaranteed to multiply to the identity, they may do so, as in (c), in which case there is no remainder.
    This is represented by $\emptyset$ in the string $y$ which denotes the remainder parity.
    For the above sequence of outputs, $y=((1),\emptyset,(1))$. This concludes Part 1 of the protocol, for Part 2 see Fig.~\ref{fig:GHZProtocolPart2}.
    }
    \label{fig:GHZProtocolPart1}
\end{figure*}

\subsection{Representative for the phase}
In order to choose a representative of the GHZ phase, we again begin with the virtual symmetry. We require the action of the symmetry in the virtual space to be given by the linear regular-representation,
\begin{align}
    U^{\rm reg}_g = \sum_{h\in G} \ket{g\circ h}\bra{h}.
\end{align}
Thus, we take $U_g^{\rm reg}$ as the corresponding physical symmetry and select the $N$-partite $|G|$-dimensional GHZ state as the representative for the phase:
\begin{align}
    \ket{{\rm GHZ}_n} = \sum_{g\in G} \ket{gg\dots g}.
\end{align}

Before presenting the protocol, we briefly summarize some useful properties of the abelian subspace of the regular representation. Let us define the operators
\begin{equation}
    Z^\alpha=\sum_{h\in G}  \chi^\alpha_{h^{-1}} \ketbra{h}, \quad X_g = \sum_{h\in G} \ketbra{h\circ g}{h},
\end{equation}
where $\alpha \in \Irr_1$ labels the 1D irreps of $G$ and $g\in G$. These operators act analogously to generalized Pauli matrices, satisfying $[U_g,X_h]=0$ and $U_g Z^\alpha = \chi^\alpha_g Z^{\alpha} U_g $ for all $g,h\in G$, $\alpha \in \Irr_1$. We then have the following observation (see Ref.~\cite{Gunn2025} for a proof).
\begin{obs} \label{obs:BasisForAbelianSubspaceOfTensorPowersOfRegularRep}
    The vectors
    \begin{equation}
    \ket*{\phi_{\alpha g_2\dots g_n}}= (Z^\alpha \otimes X_{g_2} \otimes \dots \otimes X_{g_n}) \ket{\mathrm{GHZ}_n} \label{eq:AbelianBasisRegRep},
    \end{equation}
    with $\alpha\in\Irr_1$, and $g_2,\dots,g_n\in G$, form a symmetric orthonormal basis for the abelian subspace of $(U_g^{\rm reg})^{\otimes n}$ and 
    \begin{equation}
        (U_g^{\rm reg})^{\otimes n} \ket*{\phi_{\alpha g_2\dots g_n}} = \chi^\alpha_g \ket*{\phi_{\alpha g_2\dots g_n}}.
    \end{equation}
\end{obs}

\subsection{Protocol}
We now present in detail the three parts of the protocol preparing the GHZ state from a product state. An example of the protocol for the class-4 nilpotent group $D_{32}$ is shown in Figs.~\ref{fig:GHZProtocolPart1} and~\ref{fig:GHZProtocolPart2}. The protocol proceeds in three parts.

\subsubsection{Part 1}
We begin by locally preparing, on each site, the fiducial state defined in Eq.~\eqref{eq:fiducial} of the GHZ tensor. We will refer to the left and right leg of the fiducial states as the auxiliary qudits and the central leg as the physical qudit. Then we use a depth-2 circuit to move each right auxiliary qudit to its right neighbor. Throughout this protocol, we will only measure the auxiliary qudits, never the physical qudits---those will later hold the target state. The goal of Part 1 of the protocol is to reduce the support on auxiliary qudits to the abelian subspace on which we can then perform a Bell-type measurement.

To achieve this, we proceed as follows: In each round $m \in \{1, \dots, M-1\}$, we project onto the subspaces corresponding to the irreps within the equivalence classes of \( \Irr_{M-(m-1)}/{\sim} \), following the approach used in the SPT case. Since the measurements act only on subsets of the qudits---rather than on a pure state---we do not encounter the same parity constraints observed in the SPT scenario (for further details, see Section~\ref{sec:GHZProbabilites}). However, we can still decompose each measurement outcome $\vec{x}^{(m)}$ into minimal and complete connected substrings, except possibly for the final substring, which need not be complete (as illustrated in Fig.~\ref{fig:GHZProtocolPart1}(b)). Moreover, the substrings that are complete can be again rearranged using a SWAP circuit, and further decomposed into substrings of length at most $|G_{M-m}|/|G_{M-(m-1)}|$. Similar to the SPT case, the protocol fails in round $m$ if the required circuit depth exceeds $L^{(m)}(N^{(m)})$. Prior to performing the Bell-type measurement, we must ensure an overall parity constraint, which is the goal of Part 2.

\subsubsection{Part 2}
If the protocol succeeds in Part 1, we have obtained a state in which almost all auxiliary qubits have been grouped into supersites of size $O(1)$ that have support solely on abelian subspaces, except for possibly $O(1)$ auxiliary sites at the end of the chain. These sites can have support on any equivalence class of irreps at any level of the hierarchy. For example, if the measurement outcomes in the first round did not multiply to the identity, then the last $O(|G_{M-1}|/|G_M|)$ pairs of auxiliaries will have support on irreps in non-identity equivalence classes in $\Irr_{M}/{\sim}$. In general, we can define a ``remainder" as the extent to which the measurement outcomes in each round fail to satisfy the parity constraint. This remainder can be represented by a tuple, $\vec{y}=(y_2,\dots,y_M)$, where each $y_k$ is either a string of elements of $\Irr_{k}/{\sim}$ such that no subset multiplies to the identity, or $y_k$ is the empty set, $\emptyset$, indicating that in round $M-(k-1)$ the measurements failed to satisfy the parity constraint. The goal of Part 2 is to transform the output of Part 1 to a state in which \emph{all} sites are grouped into $O(1)$ supersites with support on abelian subspaces, i.e., to a state with remainder $\vec{y}=(\emptyset,\emptyset,\dots,\emptyset)$.

To this end, we append pairs of auxiliary sites to the end of the chain and perform symmetric projective measurements to entangle the extra auxiliaries with the chain. These measurement operators can project the new auxiliary pair onto subspaces corresponding to any level of the hierarchy. As in Part 1 of the protocol, we can rearrange the sites at the end of chain and group them into supersites of size $O(1)$ that correspond to complete substrings at some level of the hierarchy. We can then measure these supersites again in order to reach lower levels of the hierarchy as in Part 1. One can iteratively append sites, perform measurements, and rearrange sites until $y=(\emptyset,\emptyset,\dots,\emptyset)$. We now present an explicit protocol that accomplishes this in $\log(N)$ rounds of measurements, each acting only on $O(1)$ sites, interleaved with constant-depth circuits, resulting in an overall depth of $O(\log(N))$.

Before presenting the protocol, let us define the following measurement
\begin{multline}
   \{ P^{(n),k} \}= \{P_{[\alpha]}^{(n)}\}_{[\alpha]\in\Irr_2/{\sim}} \\
   \bigcup_{m=3,\dots, k}\{P_{[\alpha]}^{(n)}\}_{[\alpha]\in(\Irr_{m}/{\sim})\setminus[1]_m}.
\end{multline}
This measurement extends the usual $m=2$ measurement, $\{P_{[\alpha]}^{(n)}\}_{\alpha\in \Irr_2/\sim}$, with projectors onto the non-identity equivalence classes of irreps at each level $m=k$. We can label the outcome of this measurement by $\vec{z}=(z_2,z_3,\dots,z_M)$, where $z$ has only one entry that is not the empty set and which is an element of $\Irr_m/{\sim}$ for some $m$.

Equipped with this measurement, let us now describe the protocol in Part 2, see Fig.~\ref{fig:GHZProtocolPart2} for an example. Consider the case where $\vec{y}=(y_2,y_3,\dots,y_M)\ne(\emptyset,\emptyset,\dots,\emptyset)$. First, we append the state $\ket{\phi_{00}}$, defined in Eq.~\eqref{eq:AbelianBasisRegRep}, to the auxiliary site $N$. Now we perform the projective measurement $\{ P^{(2),M} \}$ on the right-most physical site and the auxiliary system with outcome $\vec{z}=(z_2,z_3,\dots,z_M)$. In the example in Fig.~\ref{fig:GHZProtocolPart2}(a) this corresponds to $\vec{z}=(\emptyset,1,\emptyset)$ (green box). Let $z_k$ be the entry of $\vec{z}$ that is not the empty set and consider the $k^{th}$ entry of the remainder $\vec{y}$, i.e., $y_k=(y_{ki})_i$. If $z_k \circ \prod_i y_{ki}=[1]\in \Irr_{k}/{\sim}$, we can rearrange the free sites and measure again, this time with $\{ P^{(2(l(y_k)+1)),k-1}\}$, see Fig.~\ref{fig:GHZProtocolPart2}(b). This corresponds to updating the remainder as $y_k\mapsto \emptyset$ and produce a new output, $z'$, that has exactly one entry not equal to $\emptyset$, which is at entry strictly less than $k$. Similarly, if only a subset of $S\subseteq y_kz_k$ multiplies to the identity, apply the above procedure to that subset only and update $y\mapsto y_kz_k\setminus{S}$  accordingly. One can repeat this procedure until either no subset of $y_kz_k$ multiplies to the identity or $y_2=\emptyset$. If $\vec{y}=(\emptyset,\emptyset,\dots,\emptyset)$, Part 2 terminates; otherwise, we repeat. 

In this way, at the end of Part 2, the auxiliary qubits of the initial fiducial states can be partitioned into connected substrings of size $O(|G|)$ that only have support on $[1] \in \Irr_{2}/{\sim}$, i.e., the abelian subspace. Finally, we move the remaining physical site at the end of the chain back to site $N$.

\begin{figure}
    \centering
    \includegraphics[width=1\linewidth]{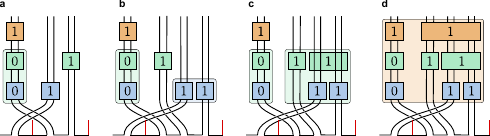}
    \caption{Example of Part 2 of the protocol, which establishes the parity constraint. At the end of Part 1 we can describe the remainder by $\vec{y}=((1),\emptyset,(1))$; see Fig.~\ref{fig:GHZProtocolPart1}(d). (a) We append auxiliary systems in the state $\ket{\phi_{00}}$ and perform the measurement $P^{(2),M}_{[\alpha]}$. The measurement outcome can belong to levels $2,3$ or $4$ of the hierarchy. Let us assume we obtained an outcome belonging to $\Irr_{2}/{\sim}$ (green box). The remainder updates to $\vec{y}=((1),(1),(1))$, where none of the $y_i$ contain a subset whose elements multiply to the identity. Therefore, in (b), we rearrange the sites and append another auxiliary system in the state $\ket{\phi_{00}}$. We measure again, obtaining the updated remainder $\vec{y} = ((1),(1),(1,1))$. Now $y_3$ multiplies to the identity, so in (c) we perform another measurement, yielding $\vec{y} = ((1),(1,1),\emptyset)$. Next, $y_2$ multiplies to the identity, and in (d) we measure once more, obtaining $\vec{y} = ((1,1),\emptyset,\emptyset)$. At this point, the total parity of the state is the identity, and we can proceed to Part~3.
}
    \label{fig:GHZProtocolPart2}
\end{figure}

\subsubsection{Part 3}
Now equipped with a state that can be decomposed into sites of size $O(|G|)$ that only have support on the abelian subspace, we project onto the orthonormal basis for said subspace given in Observation~\ref{obs:BasisForAbelianSubspaceOfTensorPowersOfRegularRep}. It can be verified, using similar arguments as in Sec.~\ref{sec:SPTProbabilities}, that any state obtained after the measurement is equivalent to the desired GHZ state up to quasi-commuting local unitaries (see Appendix~\ref{Appendix:VerificationGHZIsOutputed}). These can be corrected using local auxiliary systems~\cite{Gunn2025}.

\subsection{Resources}
Having verified that the protocol produces the desired state, let us review the resources required for the protocol. 

First, let us consider the number of auxiliary systems we need per site and the maximum number of sites we must simultaneously measure. The protocol begins with at most three qudits per site. Part 1 requires no additional auxiliaries. Part 2, requires $L^{M}$ auxiliaries. However, these are shared between $L^{M}$ sites. We will ultimately choose $L^{M}=O(\log(N))<O(N)$. Therefore, the protocol requires $O(1)$ auxiliaries per site. Additionally, every measurement implemented during the protocol acts on $O(1)$ sites.

As with the SPT case, each round in Part 1 requires $O(|G_m|L^{m})$ depth. In Part 2, we implement $L^{M}$ rounds of measurements. In each round of measurements, the number of auxiliaries that require further measurements is at most $|G_1|/|G_2|+\dots +|G_{M-1}|/|G_M|$. Therefore, between each measurement, we use $O(1)$ depth circuits. Therefore, the cumulative depth of these actions is $O(L^M)$. Finally, at the end of Part 2, we move the physical qubit back to site $N$, thereby adding another depth-$O(L^M)$ circuit. Thus, the cumulative depth of Part 2 is $O(L^M)$, and therefore the cumulative depth of the protocol is $\max_m O(\log^{(m)}(N))$.

\subsection{Probability distribution of measurement outcomes}\label{sec:GHZProbabilites}
The probability distribution for the measurements during the GHZ protocol is simpler than the SPT case. When evaluating the probability, analogous to Eq.~\eqref{eq:SPTEq2}, the unmeasured sites ensure that all $b_{ij}a_{ijk}$ terms conspire to equal identity. Consequently, one finds that the probability distribution is uniform, without any parity constraints. In particular, we have
\begin{equation}
    p(\vec{x}^{(m)}|\vec{x}^{(1)},\dots, \vec{x}^{(m-1)}) = \left(\frac{|G_{M-{m-1)}}|}{|G_{M-m}|}\right)^{N^{(m)}}.
\end{equation}
This also holds in Part 2 of the protocol as $P_{[\alpha]_{m-1}}^{(n)}=P_{[\alpha]_{m-1}}^{(n)}P_{[1]_m}^{(n)}$. For details, see Appendix \ref{Appendix:ProbabilityGHZ}.

\subsection{Probability of success}
As in the SPT case, we restrict ourselves to circuits of depth $O(\log(N))$ and show that the probability of success converges to $1$ as $N\rightarrow \infty$. As the measurement outcome in each round of Part 1 are independent and identically distributed, the first part of the argument is identical to the SPT case---the protocol fails in round $m$ if the measurement outcome contains a substring that is longer than $L^{(m)}$. As long as $L^{(m)}$ is chosen to scale as $O(\log(N))$, with a sufficiently large constant prefactor, then the probability that the protocol fails during Part 1 vanishes as $N\rightarrow\infty$. In Part 2, the sequence of measurements can be viewed  as a path on a finite Markov network, where each node corresponds to the possible value of $\vec{y}$. As the set of possible $\vec{y}$ is finite and all trajectories eventually terminate at $\emptyset$, choosing $L^{(M)}$ to grow with $N$ ensures that the failure probability converges to zero. In particular, $O(\log(N))$ is sufficient. Thus, combined with the resource analysis above, we have shown that it is possible to transform the trivial phase into the GHZ phase with asymptotically deterministic CMF for nilpotent groups.

\section{Conclusion and Discussion}

In this work, we have studied transformations of MPS under combinations of short-depth symmetric circuits and symmetric measurements, where the symmetry is with respect to a non-abelian nilpotent group. We explained that nilpotent groups have a powerful structure in their irreducible representations; namely, the lower central series of a nilpotent group can be used to define a sequence of ascending subsets of irreducible representations, $\Irr_m$, that can then be given a group structure, $(\Irr_m/{\sim},\circ)$. This group structure is central in the construction of the protocols. This means that if one performs on-site measurements on subspaces corresponding to irreps contained in elements of $\Irr_m/{\sim}$, then one can combine the sites into supersites that only have support on subspaces corresponding to irreps in $\Irr_{m-1}$. Using this, we showed that the trivial state can be asymptotically deterministically transformed into any SPT or GHZ-like state via circuits and measurements while respecting non-abelian nilpotent symmetries. Thus SPT and GHZ phases of non-abelian nilpotent groups trivialize with the inclusion of symmetric measurements.

Looking forward, it would be interesting to investigate whether phases of non-nilpotent non-abelian groups trivialize with the inclusion of symmetric measurements. This is not even clear for solvable groups - a simple class of non-abelian groups containing nilpotent groups. A direct application of our protocols to solvable groups does not work. For example, consider the non-nilpotent but solvable group $S_3$, the symmetric group on three elements. It has three irreps, $\{1,2,3\}$, of which only $3$ is non-commutative. However, $3\otimes 3\cong 1\oplus 2\oplus 3$. That is, tensor products with a non-commutative irrep of $S_3$ will always yield at least one non-commutative irrep. Therefore, combining non-commutative outcomes cannot yield a supersite with only support on abelian subspaces. It would also be interesting to go beyond finite groups.

\section{Acknowledgments}
We thank Maria Sharshunova for assistance with preparing the figures. We thank an anonymous referee of Ref.~\cite{Gunn2025} for suggesting that it may be possible to generalize to nilpotent groups. DG, TK, and BK acknowledge financial support from the Austrian Science Fund (FWF) through the grants SFB BeyondC (Grant No. F7107-N38), and P 32273-N27 (Stand-Alone Project). Furthermore, we acknowledge the BMW endowment fund. This publication has received funding from the European Union's Horizon 2020 HORIZON Research and Innovation Actions Programme under the calls HORIZON-CL4-2022-QUANTUM-02-SGA via Grant Agreement No. 101113690 (PASQuanS2.1) and HORIZON-CL4-2021-DIGITAL-EMERGING-02-10 via Grant Agreement No. 101080085 (QCFD). GS acknowledges financial support by the Deutsche Forschungsgemeinschaft (DFG, German Research Foundation) under Germany's Excellence Strategy -- EXC-2111 -- 390814868.

\appendix
\onecolumngrid

\section{Group structure of irreps of nilpotent groups}\label{App:IrrepGroups}

In this appendix, we show that the groups defined in the main text, $(\Irr_m/{\sim}, \circ)$, are indeed a groups and are isomorphic to $ G_{m-1}/G_m$. In fact, the way we will do this by first introducing a third group, $(\Irr_m/{\sim_\ast},\ast)$, and showing $G_{m-1}/G_m\cong (\Irr_m/{\sim_\ast},\ast) $. We will then show $(\Irr_m/{\sim_\ast},\ast) =(\Irr_m/{\sim},\circ)$.

\subsection{Proof of the isomorphism $(\Irr_m/{\sim_\ast},\ast) \cong G_{m-1}/G_m$}
In this section, we introduce the group $(\Irr_m/{\sim_\ast},\ast)$ and show its isomorphic to $G_{m-1}/G_m$. To begin, we define the equivalence relation $\sim_\ast$. To this end, let $\alpha\in \Irr_m$ be an irrep, and $U^\alpha_g$ a unitary representation thereof. By the definition of $\Irr_m$ (see Eq.~\eqref{eq:Irra}), we have $U_g^\alpha=\id$ for all $g\in G_m$. Thus, by definition of $G_m$ and $G_{m-1}$, for every $c\in G_{m-1}$ we have that $(U_c^\alpha)^{-1}(U_h^\alpha)^{-1}U_c^\alpha U^\alpha_h=\id$ for all $h\in G$. Therefore, by Schur's lemma, we have that
\begin{equation}
    U^\alpha_c = e^{i\tilde{\varphi}_c^\alpha}\id_{d^\alpha}  \label{eq:1DIrrepsfromSchur}
\end{equation}
for all $c\in G_{m-1}$, where $e^{i\tilde{\varphi}_c^\alpha}$ is a 1D irrep of $G_{m-1}$ with kernel containing $G_m$. The map $e^{i\varphi^\alpha}:G_{m-1}/G_m\rightarrow U(1)$ is the defined by quotienting out $G_m$:
\begin{equation}
    e^{i\varphi^\alpha_{[g]}}=e^{i\tilde{\varphi}^\alpha_g},
\end{equation}
where $[\cdot]$ denotes the canonical quotient map, $[\cdot]:G_{m-1}\rightarrow G_{m-1}/G_m$. It is easy to verify $e^{i\varphi^\alpha}$ is well defined and an 1D irrep of $G_{m-1}/G_m$. Thus, we can define a map from $\Irr_m$ to the irreps of $G_{m-1}/G_m$ by
\begin{equation}
    \tilde{\Phi}:\alpha\mapsto e^{i\varphi^\alpha}.
\end{equation}

We now show $\tilde{\Phi}$ is surjective. To this end, let $e^{i\phi}\in \Irr(G_{m-1}/G_m)$ and let $e^{i\tilde{\phi}}\in \Irr(G_{m-1})$ be defined by $e^{i\tilde{\phi}}=e^{i\phi}\circ \Pi$, where $\Pi$ is the canonical projection map. Now consider the representation of $G$ induced   \cite{Isaacs2006_CharacterTheoryFiniteGroups} by $e^{i\tilde{\phi}}\in \Irr(G_{m-1})$. That is, let $k_i\in G$ be representative for the coset $i\in G/G_{m-1}$ and define $\xi(g,i)$ and $h(g,i)$ by
\begin{align}
    g k_i &= k_{\xi(g,i)} h(g,i).
\end{align}
Note that not only is $G_{m}$ normal in $G_{m-1}$, it is also normal in $G$, and thus $G/G_{m-1}$ is well-defined. Then the induced representation $\text{Ind}^G_H e^{i\tilde{\phi}}$ is given by
\begin{equation}
    \text{Ind}^G_{G_{m-1}} e^{i\tilde{\phi}} (g) = \sum_{i\in G/[G_{a-1}]} e^{i\tilde{\phi}} (h(g,i)) \ket{\xi(g,i)}\bra{i} \label{eq:InducedRep}
\end{equation}
It is easily verified that $\text{Ind}^G_{G_{m-1}} e^{i\tilde{\phi}}$ is indeed a representation of $G$ and therefore can be decomposed into irreps of $G$. Note, we also have
\begin{equation} \label{eq:InducedRepOnNilpotentSubgroup}
    \text{Ind}^G_{G_{m-1}} e^{i\tilde{\phi}} (h) = e^{i\tilde{\phi}} (h) \id_{|G/G_{m-1}|}
\end{equation}
for all $h\in G_{m-1}$. In particular therefore, all irreps in $\text{Ind}^G_{G_{m-1}} e^{i\tilde{\phi}}$ are mapped to $e^{i\phi}$ by $\Phi$. Therefore, there is an $\alpha\in \Irr_{m}$ such that $\Phi(\alpha)=e^{i\phi}$.

As $\tilde{\Phi}$ is surjective, we can promote it to a bijective map by defining an equivalence relation $\sim_\ast$ on $\Irr_m$. Namely, we say that $\alpha\sim_\ast\beta$ if  $\tilde{\Phi}(\alpha)=\tilde{\Phi}(\beta)$. Clearly, this is an equivalence relation. We denote the equivalence class that $\alpha\in \Irr_m$ belongs to by $[\alpha]$. Then the map $\Phi: \Irr_m/{\sim_\ast}\rightarrow \Irr(G_{m-1}/G_m)$ defined by 
\begin{equation} \label{eq:PhiMap}
    \Phi([\alpha])=\tilde{\Phi}(\alpha)
\end{equation} 
is clearly a bijection, with $\Phi^{-1}(1)=\Irr_{m-1}$. 

Finally, we can equip the set of equivalence classes $\Irr_m/{\sim_\ast}$ with a group multiplication
\begin{equation}
    [\alpha]\ast [\beta]\equiv\Phi^{-1}(\Phi([\alpha])\Phi([\beta])),.
\end{equation}
where $\Phi^{-1}$ is the pre-image. This clearly ensures $\Phi$ is a group homomorphism and therefore, as its bijective, an isomorphism. We have thus proven that
\begin{equation}\label{eq:isomorphism}
(\Irr_m/{\sim_\ast},\ast)\cong\Irr\left(\frac{G_{m-1}}{G_m}\right)\cong \frac{G_{m-1}}{G_m},
\end{equation}
where the last isomorphism follows from the fact that $G_{m-1}/G_m$ is abelian and the fact finite abelian groups are isomorphic to their irreducible representations. \\

\subsection{Proof of the isomorphism $(\Irr_m/{\sim_\ast},\ast)=(\Irr_m/{\sim},\circ)$}
We now show that the group $(\Irr_{m}/{\sim},\circ)$ from the main text is in fact exactly the same as $(\Irr_{m}/{\sim_\ast},\ast)$. To begin, for convenience, let us recall the definition of $(\Irr_m/{\sim}, \circ)$. As we will be later be comparing equivalence relations, let us use $\sim_\circ$ to refer to the original equivalence relation, Eq.~\eqref{eq:EquivalenceRelation@Levelm}. That is, for $\alpha,\beta\in \Irr_m$, we have
\begin{equation}
    \beta \sim_\circ \alpha \text{ if } \exists \gamma\in \Irr_{m-1} \text{ : } \alpha\in \gamma\otimes \beta,
    \label{eq:AppEquivRel}
\end{equation}
where $\alpha\in \gamma\otimes \beta$ means that the irrep $\alpha$ appears in the decomposition of $\gamma\otimes \beta$. Group multiplication is given by 
\begin{equation}
    [\alpha]_\circ \circ [\beta]_\circ = \{\delta \in \Irr_m : \exists \alpha'\in[\alpha]_\circ, \beta'\in[\beta]_\circ : \delta\in \alpha'\otimes \beta' \}
     \label{eq:AppGroupMultiplication}
\end{equation}
We will now show
\begin{equation}
    (\Irr_m/{\sim}_\circ,\circ)= (\Irr_m/{\sim}_\ast,\ast)
\end{equation}
Before showing this let us note the following observation
\begin{obs}
    For all $\alpha\in \Irr$, there is a unique irrep $\alpha^{-1}$ such that $1\in \alpha\otimes\alpha^{-1}$, namely the conjugate representation $\alpha^{-1}=\alpha^*$. Moreover, $\Irr_m=\Irr_m^*$.
    \label{obs:ConjugateIsUniqueInv}
\end{obs}

\begin{proof}
    The first statement is a consequence of Schur's lemma. To see this explicitly, let $\alpha,\beta\in\Irr$ then $\alpha\otimes\beta\ni1$. iff there is a non-zero vector $\ket{\psi}=\psi\otimes \id \ket{\phi^+}$ such that $U^\alpha_g\otimes U^\beta_g \ket{\psi}=\ket{\psi}$ for all $g\in G$. This holds iff $U^\alpha_g \psi (U_g^\beta)^T =\psi$ which is equivalent to $U^\alpha_g \psi =\psi (U^\beta_g)^*$, for all $g\in G$. Thus by Schur's Lemma, $\beta=\alpha^*$. The statement that $\Irr_m=\Irr_m^*$ follows trivially from the definition of $\Irr_m$.
\end{proof}

Now let us show $ (\Irr_m/{\sim}_\circ,\circ)= (\Irr_m/{\sim}_\ast,\ast)$.
\begin{proof}
First, let us verify that the equivalence classes are the same, i.e., 
\begin{equation}
    \alpha \stackrel{\circ}{\sim}\beta \Leftrightarrow \alpha \stackrel{\ast}{\sim} \beta.
\end{equation}
($\Leftarrow$). If $\alpha \stackrel{\ast}{\sim}\beta$, then there is a $\gamma \in \Irr_{m-1}$ such that $\alpha\in\gamma\otimes \beta$. Now consider a unitary representation of $\alpha,\beta$ and $\gamma$ and consider the restriction to $G_{m-1}$. Recall we have 
\begin{equation}
    U^\alpha_c = e^{i\tilde{\varphi}_c^\alpha}\id_{d^\alpha}  \label{eq:1DIrrepsfromSchur}
\end{equation}
for all $c\in G_{m-1}$ and likewise for $\beta$ and $\gamma$. However, $\gamma\in \Irr_{m-1}$ and therefore $e^{i\phi^{\tilde{\gamma}}_c}=1$ and thus
\begin{equation}
    U_c^\gamma \otimes U_c^\beta = e^{i\tilde{\phi}^\beta_c} \id_{d^\gamma d^\beta}
\end{equation}
for all $c\in G_{m-1}$. As $\alpha\in \gamma\otimes\beta$, this means $e^{i\tilde{\phi}^\beta_c}=e^{i\tilde{\phi}^\alpha_c}$ and thus
\begin{align}
    \exp(i\phi^{\alpha}_{[c]})&=\exp(i\phi^{\beta}_{[c]}),
\end{align}
for all $[c]\in G_{m-1}/G_m$. Thus $\alpha\stackrel{\circ}{\sim}\beta$.

\noindent ($\Rightarrow$). Now let $\alpha \stackrel{\circ}{\sim}\beta$, i.e., $\exp(i\phi^{\beta}_{[g]})=\exp(i\phi^{\alpha}_{[g]})$. Considering $c\in G_{m-1}$, it is clear if we tensor the conjugate of $\alpha$ with $\beta$ we have that $(U^\alpha_{c})^*\otimes U^\beta_c=\id_{d^\alpha d^\beta}$. Thus, by the definition of $\Irr_{m-1}$, all irreps in the decomposition belong to $\Irr_{m-1}$, $\alpha^*\otimes\beta\subseteq \Irr_{m-1}$. Let $\gamma\in \Irr_{m-1}$ be one such irrep, i.e., $\gamma\in \alpha^*\otimes\beta$. Then we must have $1\in \alpha^*\otimes(\gamma^*\otimes \beta)$. As $\gamma^*\otimes\beta$ decomposes into irreps, i.e., $\gamma^*\otimes \beta=\oplus_i\delta_i$, and then $\alpha^*\otimes (\gamma^*\otimes \beta)\cong \oplus_i (\alpha^*\otimes \delta_i)\ni 1 $, this means, by Obs. \ref{obs:ConjugateIsUniqueInv}, at least one of the irreps $\delta_i$ must be $\alpha$. Therefore, $\alpha\in \gamma^*\otimes\beta$, with $\gamma^*\in \Irr_{m-1}$, and thus, $\alpha\stackrel{\ast}{\sim}\beta$.

Having verified that the equivalences classes are the same, let us now verify the group multiplication is the same, i.e.,
\begin{equation}
    [\alpha] \circ [\beta] = [\alpha] \ast [\beta]
\end{equation}
Note, both the left and right hand sides are sets. Thus, we show they are equivalent by showing they are subsets of one another.

\noindent ($\supseteq)$. Let $\delta\in [\alpha] \ast [\beta]$. Then $\exists \alpha'\in[\alpha],\beta'\in[\beta]$ such that $\delta \in \alpha'\otimes \beta'$. Again considering restrictions to $G_{m-1}$, it is therefore clear that
\begin{equation}
    \exp(i\phi^{\gamma}_{[g]})=\exp(i\phi^{\alpha}_{[g]})\exp(i\phi^{ \beta}_{[g]})
\end{equation}
and, thus, $\gamma\in[\alpha]\circ[\beta]$.

\noindent $(\subseteq)$. Let $\gamma\in [\alpha]\circ[\beta]$. That is $\exp(i\phi^{\gamma}_{[g]})=\exp(i\phi^{\alpha}_{[g]})\exp(i\phi^{\beta}_{[g]})$. Therefore, we have $\gamma^*\otimes(\alpha\otimes\beta)\subseteq \Irr_{m-1}$. Let $\delta\in \gamma^*\otimes(\alpha\otimes\beta)$. Thus, $1\in \gamma^*\otimes ((\delta^*\otimes \alpha)\otimes \beta)$. Note, by definition $\delta^*\otimes \alpha\in[\alpha]$. Therefore, defining $\alpha' =\delta^*\otimes \alpha$, by Obs \ref{obs:ConjugateIsUniqueInv}, we have $\gamma\in \alpha'\otimes \beta$. Thus $\gamma\in [\alpha]\ast [\beta]$.
\end{proof}

Thus, $(\Irr_m/{\sim},\circ)$ is a group isomorphic to $G_{m-1}/G_m$. Considering the definitions in Eq.~\eqref{eq:AppEquivRel} and Eq.~\eqref{eq:AppGroupMultiplication}, it is clear that $[1]_m=\Irr_{m-1}$ is the group identity element. Moreover, by Observation \ref{obs:ConjugateIsUniqueInv} we have that $[\alpha]^{-1}=[\alpha]^*$.

\section{Example of the group structure of irreps of nilpotent groups}\label{app:Example}
In this appendix, we illustrate the structure of the groups $(\Irr_m/\sim_,\circ)$. As an example, we consider $G=D_{16}$, the dihedral group of order 16. By examining the successive commutator groups, $G_{m+1}=[G_m,G]$, we obtain the lower central series
\begin{equation}
    D_{16} \unrhd \mathds{Z}_4 \unrhd \mathds{Z}_2 \unrhd 1,
\end{equation}
and therefore, $G$ is a class-3 nilpotent group. The abelian factor groups, $G_{m-1}/G_{m}$, representing the group structures of $(\Irr_m/\sim_,\circ)$, are therefore given by $K_4, \mathds{Z}_2$, and $\mathds{Z}_2$. 

\subsection{The hierarchy of irreps}
Let us begin by constructing the sets, $\Irr_m$, of irreps in Eq.~\eqref{eq:HierachyIrr} in the main text. Recall, that $\Irr_m$ consists of those irreps that act trivially for all $g\in G_m$, see Eq.~\eqref{eq:Irra}. Equivalently, it contains those irreps whose kernels contain $G_m$, with the kernel of an irrep being defined by $\mathrm{ker}(\alpha)=\{g\in G : U_{g}^{\alpha}=\id \}$. Those kernels, and subsequently the sets $\Irr_m$, can be identified from the below character table of $D_{16}$. Here, the rows, labeled by $\alpha$, correspond to the different irreps and the columns correspond to the conjugacy classes of $G$~\cite{note6}.
\begin{equation}
    \centering
    \begin{tabular}{c||c|c|c|c|c|c|c}
       $\alpha$ & \textit{1} & \textit{2a} & \textit{2b} & \textit{2c} & \textit{4} & \textit{8a} & \textit{8b}  \\ \hline
        1 & 1& 1& 1& 1& 1& 1& 1 \\
        2 & 1& 1& -1& 1& 1& -1& -1 \\
        3 & 1& 1& 1& -1& 1& -1& -1 \\
        4 & 1& 1& -1& -1& 1& 1& 1 \\
        5 & 2& 2& 0& 0& -2& 0& 0 \\
        6 & 2& -2& 0& 0& 0& $\sqrt{2}$& $-\sqrt{2}$ \\
        7 & 2& -2& 0& 0& 0& $-\sqrt{2}$& $\sqrt{2}$ \\
    \end{tabular}
\end{equation}
The kernel of an irrep $\alpha$ is given by those elements for which $\chi(g)=\chi(1)$. Then, to check if a kernel contains $G_m$, we express $G_m$ in terms of conjugacy classes
 \begin{align}
     Z_4 &\cong \textit{1} \cup \textit{2a} \cup \textit{4}\\
     Z_2 &\cong \textit{1} \cup \textit{2a}\\
     1 &\cong \textit{1}.
 \end{align}
By inspecting the character table it is then easily verified that the corresponding hierarchy of irreps is given by
\begin{align}
    \Irr_0&=\{1\}\\
    \Irr_1&=\{1,2,3,4\}\\
    \Irr_2&=\{1,2,3,4,5\}\\
    \Irr_3&=\{1,2,3,4,5,6,7\}=\Irr.
\end{align}

\subsection{Group structure}
Next, to identify corresponding equivalence classes we need to consider tensor products of irreps and their decompositions. We find the following:
\begin{equation}
    \centering
    \begin{tabular}{c||c|c|c|c|c|c|c}
        $\otimes $ & 1 & 2 & 3 & 4 & 5 & 6 & 7  \\ \hline \hline
        1 & 1 & 2 & 3 & 4 & 5 & 6 & 7 \\ \hline
        2 & 2 & 1 & 4 & 3 & 5 & 7 & 6 \\ \hline
        3 & 3 & 4 & 1 & 2 & 5 & 7 & 6 \\ \hline
        4 & 4 & 3 & 2 & 1 & 5 & 6 & 7 \\ \hline
        5 & 5 & 5 & 5 & 5  & $1\oplus2\oplus3\oplus4$ & $6\oplus 7$ & $6\oplus 7$ \\ \hline
        6 & 6 & 7 & 7 & 6 & $6\oplus 7$ & $1\oplus 4\oplus 5$ & $2\oplus3\oplus 5$ \\ \hline
        7 & 7 & 6 & 6 & 7 & $6\oplus 7$ & $2\oplus3\oplus5$  & $1\oplus 4\oplus 5$ \\ 
    \end{tabular}
\end{equation}
To form equivalence classes in $\Irr_m$ under the equivalence relation stated in Eq.~\eqref{eq:EquivalenceRelation@Levelm} in the main text, we consider the irreps in $\Irr_m$ and their tensor products by irreps in $\Irr_{m-1}$. For instance, for $m=3$, $\Irr_m=\Irr$. Consider the irreps $5$ and $1$. Clearly, there is an irrep, $\gamma$, in $\Irr_{m-1}=\{1,2,3,4,5\}$ that such that $\gamma\otimes 5\ni 5$ (namely 1). Conversely, $5\in\Irr_{m-1}$ and $5\otimes 5\ni 1$. Thus, $1\sim 5$. Continuing this analysis, yields a partitioning of each $\Irr_m$ by the equivalence class described above, yielding
\begin{align}
    \Irr_0/{\sim} &=\{\{1\}\}\\
    \Irr_1/{\sim}&=\{\{1\},\{2\},\{3\},\{4\}\}\\
    \Irr_2/{\sim}&=\{\{1,2,3,4\},\{5\}\}\\
    \Irr_3/{\sim}&=\{\{1,2,3,4,5\},\{6,7\}\}.
\end{align}

Now, to see the group structure, consider as an example $\Irr_3/\sim$. According to the above analysis, we should have $\Irr_3/{\sim} \cong \mathds{Z}_2$. 
To see this consider tensor products of irreps. We have 
\begin{align}
    1\otimes6=4\otimes 6= 2\otimes 7=3\otimes 7=6 \in &\{6,7\}\\
    1\otimes 7=4\otimes7=2\otimes6=3\otimes 6=7\in&\{6,7\}\\
    5\otimes 5 = 6\oplus 7 \subseteq &\{6,7\}.
\end{align}
Thus, we can see that $\{1,2,3,4,5\}\circ\{6,7\}=\{6,7\}$. It is similarly verified that $\{1,2,3,4,5\}\circ\{1,2,3,4,5\}=\{1,2,3,4,5\}$ and $\{6,7\}\circ\{6,7\}=\{1,2,3,4,5\}$. Then, $\Irr_3/\sim$ is indeed isomorphic to $\mathds{Z}_2$. Similarly, it is easy to verify that $\Irr_1/{\sim}\cong G_0/G_1\cong K_4$ and $\Irr_2/{\sim}\cong G_1/G_2\cong \mathds{Z}_2$. Moreover, we can verify that $[1]_m=\Irr_{m-1}$ acts as the identity element for each $m$.

\subsection{Projectors}
Finally, let us continue with $G=D_{16}$ and consider a example of Eq.~\eqref{eq:projectorRule}, which we restate for convenience:
\begin{equation}
    P_{[\alpha]_m}^{(n)}\otimes P_{[\beta]_m}^{(n')} \subseteq P_{[\alpha]_m \circ [\beta]_m}^{(n + n')}.
\end{equation}
To simplify things, consider an onsite symmetry $U_g=1 \oplus 6\oplus 7$ (the following also holds for any representation, including the regular representation, but demonstrating the notation for larger examples is tedious). Then 
\begin{align}
   P_{[1]_M}^{(1)}&= (1)_1\oplus (0_2)_6\oplus (0_2)_7 \\
   P_{[2]_M}^{(1)}&=(0)_1\oplus (\id_2)_6\oplus (\id_2)_7,
\end{align}
where the subscripts label the irreps corresponding to the subspace, and $0_2$ is shorthand for a $2\times 2$ zero matrix. Then
\begin{align}
    P_{[2]_M}^{(1)}\otimes P_{[2]_M}^{(1)}
    &\ \ = (0 \otimes 0)_{1\otimes 1} \oplus (0 \otimes \id_2)_{1\otimes 6} \oplus (0 \otimes \id_2)_{1\otimes 7} \\
    &\quad  \ \ \ (\id_2 \otimes 0)_{6\otimes 1} \oplus (\id_2 \otimes \id_2)_{6\otimes 6} \oplus (\id_2 \otimes \id_2)_{1\otimes 7} \\
    &\quad \ \ \  (\id_2 \otimes 0)_{7\otimes 1} \oplus (\id_2 \otimes \id_2)_{7\otimes 6} \oplus (\id_2 \otimes \id_2)_{7\otimes 7} \\
    &\ \ \subseteq (1)_{1\otimes 1} \oplus (0_2)_{1\otimes 6} \oplus (0_2)_{1\otimes 7} \\
    &\quad  \ \ \ (0_2)_{6\otimes 1} \oplus (1\oplus 1\oplus \id_2)_{6\otimes 6} \oplus (1\oplus 1\oplus \id_2)_{1\otimes 7} \\
    &\quad \ \ \  (0_2)_{7\otimes 1} \oplus (1\oplus 1\oplus \id_2)_{7\otimes 6} \oplus (1\oplus 1\oplus \id_2)_{7\otimes 7} \\
    &\ \ =P^{(2)}_{[1]_M},
\end{align}
where in the third line we have decomposed $6\otimes 6 =7\otimes 7 =1\oplus 4\oplus 5$ and $6\otimes 7=7\otimes 6=2\oplus 3\oplus 5$. Thus, $ P^{(1)}_{[2]_M}\otimes P^{(1)}_{[2]_M}$ only has support on subspaces corresponding to irreps in $\Irr_{M-1}=[1]_M=[2]_M\circ[2]_M$. Note, that $P_{[2]_M}^{(1)}\otimes P_{[2]_M}^{(1)}\ne P_{[1]_M}^{(2)}$.

\section{Proof of Lemma~\ref{lem:IrrepSumTrick}}\label{app:SumProof}
In this section, we prove Lemma~\ref{lem:IrrepSumTrick}, which we restate here for convenience. 
\setcounter{thm}{0}
\begin{lem}
Let $\Irr_m/{\sim}\cong G_{m-1}/G_m=\{[\alpha]\}$, and $\Phi:\Irr_m/{\sim}\rightarrow \Irr(G_{m-1}/G_m)$ as defined in Eq.~\eqref{eq:PhiMap} with $\Phi([\alpha])=e^{i\varphi^{[\alpha]}}$. Then,
\begin{equation}
    P^{(n)}_{[\alpha]}= \sum_{g\in G_{m-1}} \frac{e^{-i\varphi^{[\alpha]}([g])}}{\abs{G_{m-1}}} U_g^{\otimes n}.
\end{equation}
\end{lem}

\begin{proof}
    Let $[\alpha]\in \Irr_m/\sim \ \cong G_{m-1}/G_m$. We have, by Eq.~\eqref{eq:ProjectorOntoIrrepClasses} in the main text, that
    \begin{equation}
        P_{[\alpha]}^{(n)}=\sum_{\alpha\in [\alpha]}P_\alpha^{(n)}=  \sum_{\alpha\in [\alpha]} \sum_{g\in G} \frac{d^{\alpha} \bar{\chi}^{\alpha}_g}{|G|} U_g^{\otimes n}.
    \end{equation}
    Let $e^{i\tilde{\phi}^{\alpha}}\in \Irr(G_{m-1})$ be defined by $e^{i\phi^{[\alpha]}}\equiv e^{i\tilde{\phi}^{\alpha}}\circ \Pi$, where $\Pi$ is the canonical projection from $G_{m-1}$ to $G_{m-1}/G_{m}$. We will prove the Lemma by proving the following:
    \begin{equation}
        \sum_{\alpha\in [\alpha]}\frac{d^\alpha\bar{\chi}^\alpha_g}{\abs{G}}=\frac{e^{-i\varphi^{[\alpha]}([g])}}{\abs{G_{m-1}}}\delta_{g\in G_{m-1}}. \label{eq:LemmaEq}
    \end{equation}
    The proof proceeds by relating both sides of the above equation to the character of the induced representation ${\rm Ind}^G_{G_{m-1}}  e^{i \tilde{\phi}^{\alpha}}$ (cf. Eq.~\eqref{eq:InducedRep}). It is easy to verify that its character is given by
    \begin{equation}
        \tr[{\rm Ind}^G_{G_{m-1}}  e^{i \tilde{\phi}^{\alpha}} (g)]= e^{i \phi^{[\alpha]}} \frac{|G|}{|G_{m-1}|} \delta_{g\in G_{m-1}}.
    \end{equation}
    Another way of expressing the character can be obtained as follows: Due to Eq.~\eqref{eq:InducedRepOnNilpotentSubgroup}, all irreps in ${\rm Ind}^G_{G_{m-1}}  e^{i \tilde{\phi}^{\alpha}}$ belong to $[\alpha]$.
    For any $\alpha' \in [\alpha]$ we can compute its multiplicity, $m^{[\alpha]}(\alpha')$, in the induced representation ${\rm Ind}^G_{G_{a-1}}  e^{i \tilde{\phi}^{\alpha}}$. By simple character theory, we have that
    \begin{align}
        m^{[\alpha]}(\alpha')&= \frac{1}{|G|} \sum_{g\in G} \bar{\chi}_g^{\alpha'} \tr[{\rm Ind}^G_{G_{m-1}}  e^{i \tilde{\phi}^{\alpha}}(g)]\\
        &= \frac{1}{|G|} \sum_{g\in G} \bar{\chi}_g^{\alpha'} \left[ e^{i \phi^{[\alpha]}} \frac{|G|}{|G_{m-1}|} \delta_{g\in G_{m-1}} \right]\\
        &= \frac{1}{|G_{m-1}|} \sum_{g\in G_{m-1}} d^{\alpha'} e^{-i \phi^{[\alpha]}_g} [e^{i \phi^{[\alpha]}_g}]\\
        &= d^{\alpha'}
    \end{align}
    That is, every irrep $\alpha'\in[\alpha]$ appears in ${\rm Ind}^G_{G_{m-1}}  e^{i \tilde{\phi}^{\alpha}}(g)$ with multiplicity equal to its dimension. Therefore, we can write
    \begin{equation}
        {\rm Ind}^G_{G_{m-1}}  e^{i \tilde{\phi}^{\alpha}}(g) \cong \bigoplus_{\alpha'\in [\alpha]} (U^{\alpha'}_g)^{\oplus d^{\alpha'}},
    \end{equation}
    where $U^{\alpha}_g$ is some canonical unitary representation of the irrep $\alpha$. Thus, we can also express the character as
    \begin{align}
        \tr[{\rm Ind}^G_{G_{m-1}}  e^{i \tilde{\phi}^{\alpha}} (g)]&= \sum_{\alpha' \in [\alpha]} d^{\alpha'} \chi_g^{\alpha'} = e^{i \phi^{[\alpha]}} \frac{|G|}{|G_{m-1}|} \delta_{g\in G_{m-1}}.
    \end{align}
    Taking the conjugate of both sides of this last equality, we have Eq.~\eqref{eq:LemmaEq} and hence the lemma.
\end{proof}

\section{GHZ protocol: Verification of the final state}\label{Appendix:VerificationGHZIsOutputed}
In this appendix, we show that the protocol in Section~\ref{sec:GHZProtocol} outputs a state locally symmetric-unitary equivalent to the desired GHZ state. For simplicity, let us consider a class-3 nilpotent group. It will be clear how the construction extends to larger classes. Additionally, let us consider the case where no Part 2 is required and no circuits are needed; that is, in Part 1, the measurement outcomes conspire so that each measurement output, $\vec{x}^{(m)}$, already decomposes into complete substrings of length at most $|G_{M-m}|/|G_{M-(m-1)}|$. We will address the effect of Part 2 and permutation circuits at the end of the section. 

Due to this assumption, we may write the outcomes of the first two measurement rounds as 
\begin{align}
    \vec{x}^{(2)}&=(\vec{x}^{(2)}_i)_{i=1}^{N^{(3)}}=((x^{(2)}_{ij})_{j=1}^{N^{(2)}_i})_{i=1}^{N^{(3)}} \\
    \vec{x}^{(1)}&=(\vec{x}^{(1)}_i)_{i=1}^{N^{(3)}}=((\vec{x}^{(1)}_{ij})_{j=1}^{N_i^{(2)}})_{i=1}^{N^{(3)}}=(((x^{(1)}_{ijk})_{k=1}^{N_{ij}^{(1)}})_{j=1}^{N_i^{(2)}})_{i=1}^{N^{(3)}}
\end{align}
where (see Fig.~\ref{fig:GHZOutputFirstBlock})
\begin{align}
    \circ_k x_{ijk}^{(1)}&=[1]\in \Irr_{M=3}/\sim,\qquad \ \ \forall ij\\
    \circ_j x_{ij}^{(2)}&=[1]\in \Irr_{(M-1)=2}/\sim,\quad \forall i.
\end{align}
Note, that we use the upper index to refer to the measurement round, and the lower multi-index to refer to the partitioning (and subpartitioning) of the string of measurement outcomes according to how they multiply to the identity. The lower index is a multi-index, so, for example, $x^{(1)}_{11N^{(1)}_{11}}$ is followed by $x^{(1)}_{121}$, which is followed by $x^{(1)}_{122}$, and so on.

Then, in the last round of the protocol, on supersites corresponding to each $\vec{x}_i^{(2)}$ for each $i\in\{1,\dots,N^{(3)}\}$, we project the 1D subspaces of $U_g^{\otimes n'}$, where $n'=\sum_j 2N_{ij}$ (see Fig.~\ref{fig:GHZOutputFirstBlock}). The eigenvectors of these 1D subspaces are given in Obs.~\ref{obs:BasisForAbelianSubspaceOfTensorPowersOfRegularRep}, which we restate for ease of reading. First, we define $Z^\alpha=\sum_{h\in G}  \chi^\alpha_{h^{-1}} \ketbra{h}$ and $X_g = \sum_{h\in G} \ketbra{h\circ g}{h}$. Recall, $Z^\alpha$ and $X_g$ are quasi-commuting with $U_g$.  Then the common eigenvectors of $U_g^{\otimes n}$ are given by  
\begin{equation}
    \ket*{\phi_{\alpha g_2\dots g_n}}= (Z^\alpha \otimes X_{g_2} \otimes \dots \otimes X_{g_n}) \ket{\mathrm{GHZ}_n}.
\end{equation}
Recall, throughout the protocol, we only every measure on the auxiliary qudits of the fiducial state. Therefore $g_{ijkL}$ corresponds to the left auxiliary qudit and $g_{ijkR}$ corresponds to the right auxiliary qudit. Thus, using the notation introduced above, in the last round we project onto 
\begin{equation}
    |\alpha_i, g_{i11R}, g_{i12L}, g_{i12R},\dots, g_{iN_i^{(2)}N_{iN_i}^{(1)}L}, g_{iN_i^{(2)}N_{iN_i}^{(1)}R}\rangle
   \label{eq:AppendixAbelianBasisRegRep}
\end{equation}
for each $i\in\{1,\dots,N^{(3)}\}$ (see Fig.~\ref{fig:GHZOutputFirstBlock}).

\begin{figure}
\centering
\includegraphics[width=0.4\linewidth]{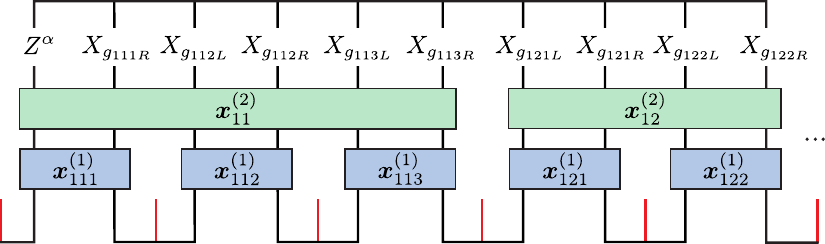}
\caption{ 
    Example of projective measurements on the first substring of sites in round $3$ for a class-3 Nilpotent group (assuming no permutation circuits are required). Let us discuss this image as an aid to understanding the notation introduced. Looking at this image, we can read off the measurement history of the protocol. Starting from the top of the image, we see the final projective measurement, corresponding to Eq.~\eqref{eq:AppendixAbelianBasisRegRep}. Moving down, there were two projective measurements (green boxes) in round two. Thus, these two measurements must have totaled to identity; i.e., $N_1^{(2)}=2$ and $x^{(2)}_{11}\circ\dots\circ x^{(2)}_{1N^{(2)}_1}=[1]$. Looking under the first round-two-measurement, $x^{(2)}_{11}$, we see three round one measurements (blue boxes). Thus, these measurements must have totaled to identity; i.e., $N^{(1)}_{11}=3$ and $ x_{111}^{(1)}\circ\dots\circ x_{11N^{(1)}_{11}}^{(1)}=[1]$. Note, this is just the first substring ($i=1$); the state continues to the right.     
}
\label{fig:GHZOutputFirstBlock}
\end{figure}

We will now show that the state after this projection is, up to on-site quasi-commuting local unitaries, the GHZ state. As we consider the case where no circuits are needed, we may consider the qudits associated with each substring $x^{(2)}_i$ in turn. Let us evaluate the first ($i=1$) substring, an example of which is depicted in Fig.~\ref{fig:GHZOutputFirstBlock}. To this end, we use decomposition of the projectors $P^{(n)}_{[\alpha]}$ at level $m$ as given in  Lemma~\ref{lem:IrrepSumTrick},
\begin{equation}
    P^{(n)}_{[\alpha]}= \sum_{g\in G_{m-1}} \frac{e^{-i\varphi^{\alpha}([g])}}{\abs{G_{m-1}}} U_g^{\otimes n},
    \label{eq:appendixprojectorequation}
\end{equation}
where $U_g$ is the regular representation. Evaluating Fig.~\ref{fig:GHZOutputFirstBlock} yields
\begin{multline}
     \sum_{\substack{\{b_{j}\in G_{M-2}\}_{j}\\
     \{a_{jk}\in G_{M-1}\}_{jk}}} 
        \left( 
            \prod_{j=1}^{N_1^{(2)}} \frac{e^{-i \phi^{x_{1j}^{(2)}}_{[b_{j}]}}}{|G_{M-2}|} 
            \left( 
                 \prod_{k=1}^{N^{(1)}_{1j}} \frac{e^{-i \phi^{x_{1jk}^{(1)}}_{[a_{jk}]}}}{|G_{M-1}|} 
            \right)
        \right)
      \times
      \begin{aligned}
        \includegraphics[width=0.5\linewidth]{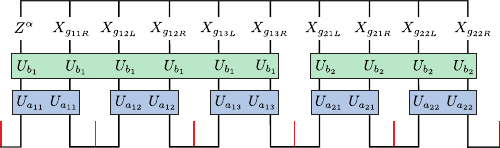}
     \end{aligned}
    \label{eq:GHZoutput0a}
\end{multline}
Here, we use colors to indicate the origin of each object. We denote by $b_j$ the group elements obtained from decomposing the measurements outcomes of the second round via Eq.~\eqref{eq:appendixprojectorequation}, and by $a_{jk}$ those arising from the measurement outcomes in the first round. Inspecting the graphical object in Eq.~\eqref{eq:GHZoutput0a}, we see that the last term is a product of objects of the form
\begin{equation}
    \begin{aligned}
        \includegraphics[width=0.14\linewidth]{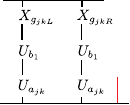}
    \end{aligned}
    \label{Eq:GHZbridge}
\end{equation}
As $U_g$ is the regular representation and $X_g$ is the right-regular representation, we can evaluate this object as
\begin{equation}
   \sum_{e, d_{jk}, d_{jk\ominus1}\in G} \delta[e=b_{j}a_{jk}d_{jk\ominus 1}g_{1jkL} ] 
       \delta[e=b_{j}a_{jk}d_{jk}g_{1jkR} ]
  \ketbra{e}\otimes \ket{d_{jk\ominus1}}\otimes \ket{d_{jk}} \otimes \ket{d_{jk}}.
\end{equation}
Here, we have introduced the notation $\ominus$ to indicate the previous value of the multi-index (e.g, $(1,2)\ominus1 = (1,1)$ and $(2,1)\ominus 1= (1,N_{11})$). The graphical object in Eq.~\eqref{eq:GHZoutput0a} is a product of \eqref{Eq:GHZbridge} terms, with exception of the first term which required additional care. Firstly, according to the above notation, the first term will also contain $d_{11\ominus 1}$. To simplify, we write this as $d_{11\ominus 1} = d$. Secondly, we have a $Z^{\alpha_1}$ instead of a $X_g$. Therefore, we define $X_{g11L}=\id$ and the $Z^{\alpha_1}$ will yield the phase  $\chi^{\alpha_1}_{b_1a_11 d}$. Putting this all together, Eq.~\eqref{eq:GHZoutput0a} becomes
\begin{multline}
    \sum_{\substack{\{b_{j}\in G_{M-2}\}_{j}\\
    \{a_{jk}\in G_{M-1}\}_{jk}}} 
    \left( 
    \prod_{j=1}^{N_1^{(2)}} \frac{e^{-i \phi^{x_{1j}^{(2)}}_{[b_{j}]}}}{|G_{M-2}|} 
    \left( 
    \prod_{k=1}^{N^{(1)}_{1j}} \frac{e^{-i \phi^{x_{1jk}^{(1)}}_{[a_{jk}]}}}{|G_{M-1}|} 
    \right)
    \right)\\
    \sum_{\substack{\{d_{jk}\in G\}_{jk}\\d,e\in G}} \chi^{\alpha_1}_{b_{1}a_{11}d} 
    \Bigg(
    \prod_{j=1}^{N^{(2)}_{1}}\prod_{k=1}^{N^{(1)}_{1j}}  \delta[e=b_{j}a_{jk}d_{jk\ominus 1}g_{1jkL}]
    \delta[e=b_{j}a_{jk}d_{jk}g_{1jkR} ]\Bigg)
    |d\rangle\langle d_{N^{(2)}_1  N^{(1)}_{1,N^{(2)}_1}}| \otimes \bigotimes_{jk} \ket{d_{jk}}.
    \label{eq:GHZoutput1}
\end{multline}
As $\chi^{\alpha_1}$ is a 1D irrep of $G$ with kernel $G_1\unrhd G_{m\ge1}$, we have that $\chi^{\alpha_i}_{b_{1}a_{11}d}=\chi^{\alpha_i}_{d}$. Next, note the LHS and first two terms on the RHS of each delta function constraint are identical. Therefore, rearranging terms we get the recursive conditions
\begin{align}
    d_{jk\oplus1} = d_{jk}g_{1jk},
\end{align}
where we define $(g_{1jkL})(g_{1jkR})^{-1}=g_{1jk}$. Therefore, we have
\begin{align}
    d_{jk} = d \left( \prod_{pq=111}^{jk} g_{1pq}\right) \equiv d_{d,1jk}.
\end{align}
Note, in a slight abuse of notation, $d_{d,1jk}$ has $d$ in its indices to indicate that it is determined by $d$ (which we are still summing over). Therefore, summing over $d_{jk}$ and $e$, Eq.~\eqref{eq:GHZoutput1} simplifies to
\begin{multline}
    \sum_{\substack{\{b_{j}\in G_{M-2}\}_{j}\\
    \{a_{jk}\in G_{M-1}\}_{jk}}} 
    \left( 
    \prod_{j=1}^{N_1^{(2)}} \frac{e^{-i \phi^{x_{1j}^{(2)}}_{[b_{j}]}}}{|G_{M-2}|} 
    \left( 
    \prod_{k=1}^{N^{(1)}_{1j}} \frac{e^{-i \phi^{x_{1jk}^{(1)}}_{[a_{jk}]}}}{|G_{M-1}|} 
    \right)
    \right) \\
    \sum_{\substack{\{d_{jk}\in G\}_{jk}\\d,e\in G}}
    \Bigg(
    \prod_{j=1}^{N^{(2)}_{1}}\prod_{k=1}^{N^{(1)}_{1j}}  \delta[e=b_{j}a_{jk}d_{d,1jk}g_{1jkR} ] \Bigg) Z^{\alpha_i}|d\rangle\langle  d| X_{\left( \prod_{pq} g_{1pq}\right)}^\dagger \otimes  \bigotimes_{jk} X_{\left( \prod_{pq=11}^{jk} g_{ipq}\right)}  |d \rangle.
    \label{eq:GHZoutput2}
\end{multline}
Here, the last term contains a GHZ-like contribution (up to quasi-commuting $X_g$ and $Z^\alpha$ operators). However, the delta functions still couple the $d$ variables associated with the GHZ term to the $b_j$ and $a_{ij}$ variables, which must be summed over. We can rearrange the remaining delta-functions to obtain
\begin{equation}
    \delta\Big[b_{j}a_{jk} d_{d,1jk} g_{1jkR}=\nonumber b_{(jk\oplus1)_{1}}a_{jk\oplus1}d_{d,1(jk\oplus1)} g_{1(jk\oplus 1)R} \Big],
    \label{eq:GHZoutput3}
\end{equation}
where $\oplus$, in the same spirit as $\ominus$, indicates moving to the next multi-index, e.g., $(1,1)\oplus1=(1,2)$. Now we distinguish two cases, depending on whether $jk \mapsto jk\oplus 1$ crosses a substring boundary of $x^{(1)}$. First, suppose it does not; i.e., assume that $(jk\oplus1)=j(k\oplus1)$. Thus, Eq.~\eqref{eq:GHZoutput3} becomes
\begin{align}
    &\delta\Big[a_{jk} d_{d,1jk} g_{1jkR} =a_{j(k\oplus1)} d_{d,1j(k\oplus1)} g_{1j(k\oplus 1)R} \Big] \\
    =\ &\delta[d_{d,1jk}^{-1}a_{j(k\oplus1)}^{-1}a_{jk}d_{d,1jk} = g_{1j(k\oplus1)L}g_{1jkR}^{-1}],
\end{align}
where we have used $d_{d,1(jk\oplus1)}=d_{d,ijk}g_{1j(k\oplus1)}$ and $g_{1j(k\oplus1)}=g_{1j(k\oplus1)L}g_{1j(k\oplus1)R}^{-1}$. As $a_{jk}\in G_{M-1}\unlhd G$, the above condition gives us a constraint that $g_{1j(k\oplus1)L}g_{1jkR}^{-1}\in G_{M-1}$. 
Moreover, we may eliminate $d_{d,1jk}\in G$ (which recall, depends on $d$) by noting that as $a_{jk}\in G_{m-1}$, it commutes with all elements in $G$. 
Thus, we have
\begin{align}
    &= \delta[a_{j(k\oplus1)}=a_{jk} \tilde{g}_{1jk} ]\delta[\tilde{g}_{1jk}\in G_{M-1}]\nonumber\\
    &= \delta\left[a_{jk}=a_{j1}  \left(\prod_{q=1}^{k} \tilde{g}_{1jq} \right)\right] \delta[\tilde{g}_{1jk}\in G_{M-1}],
\end{align}
for $k\in \{2,\dots,N_{1j}^{(2)}\}$ where we have defined $\tilde{g}_{1jk}=g_{1jkR}g_{1j(k\oplus1)L}^{-1}$.

We now do the same considering when $jk \mapsto jk\oplus 1$ crosses a substring boundary of $x^{(1)}$, i.e., $k=N_{1j}^{(2)}$ and $(jk\oplus1)=(j\oplus1),1$. This time Eq.~\eqref{eq:GHZoutput3} reduces to
\begin{align}
    &\delta\Big[b_{j}a_{jN_{1j}^{(2)}} d_{d,1jN_{1j}^{(2)}} g_{1jN_{1j}^{(2)}R}= b_{j\oplus1}a_{(j\oplus1) 1}d_{d,1(j\oplus1)1} g_{1(j\oplus 1)1R} \Big] \nonumber \\
    = \ & \delta\Big[ d_{d,1jN_{1j}^{(2)}}^{-1} a_{(j\oplus1) 1}^{-1} b_{j\oplus1}^{-1} b_{j}a_{jN_{1j}^{(2)}} d_{d,1jN_{1j}^{(2)}} =  g_{1(j\oplus 1)1L} g_{1jN_{1j}^{(2)}R}^{-1} \Big].
\end{align}
Thus, this time we get the constraint $\tilde{g}_{{1(j\oplus 1)1}}\in G_{m-2}$. Moreover, for all $g\in G$ and $b\in G_{m-2}$, there is some $a_{b,g}\in G_{m-1}$ such that we have $g b= b g a_{b,g}$. Therefore, rearranging, we have
\begin{align}
    &b_{(j\oplus1)}=b_{j} \tilde{g}_{1(j\oplus1)1} \tilde{a}_{d,1jk}\\
    \Rightarrow\  & b_{j}=b_{1} \left(\prod_{p=1}^j  \tilde{g}_{1p1} \right)\tilde{\tilde{a}}_{d,1jk} \nonumber\\
    \Rightarrow \ & [b_{j}]=[b_{1} \left(\prod_{p=1}^j  \tilde{g}_{1p1} \right)],
\end{align}
for some $\tilde{a}_{d,1jk}, \tilde{\tilde{a}}_{d,1jk}\in G_{M-1}$ which depend on $j,k$ and $d$. Note, however, in taking the quotient map, we lose the dependence on $d$. 
Putting this all together, Eq.~\eqref{eq:GHZoutput2} becomes
\begin{multline}
    \left(\prod_{j\ge 2} \delta[\tilde{g}_{1j1}\in G_{M-2}]\right) \left(\prod_{jk\ne{j\ge2,1}} \delta[\tilde{g}_{1jk}\in G_{M-1}]\right)\\ 
    \left( 
    \prod_{j=1}^{N^{(2)}_{1}} \frac{e^{-i \phi^{x_{1j}^{(2)}}_{[\prod_{p=1}^j  \tilde{g}_{1p1} ]}}}{|G_{M-2}|} 
    \Bigg(\prod_{k=1}^{N^{(1)}_{1j}} \frac{e^{-i \phi^{x_{1jk}^{(1)}}_{\left[ \prod_{q=1}^{k} \tilde{g}_{1jq} \right]}}}{|G_{M-1}|} \Bigg)
    \right)
    \sum_{d\in G}  Z^{\alpha_i}|d\rangle\langle  d|  
    X_{\left( \prod_{pq} g_{1pq}\right)}^\dagger
    \otimes  \bigotimes_{jk} X_{\left( \prod_{pq=11}^{jk} g_{ipq}\right)}  |d \rangle.
    \label{eq:GHZoutput4}
\end{multline}
That is, for each substring, $\vec{x}_{i}^{(2)}$, we get a segment of the desired GHZ state up to proportionality and quasi-commuting local unitaries.  

The same argument holds for the substrings associated with Part 2.
The only difference is there is no physical site associated with the \eqref{Eq:GHZbridge} terms.
In the Part 2 protocol, it is possible to ``skip'' measurements associated with a given level of the irrep hierarchy (see measurements associated with the first level of the irrep hierarchy in last image of Fig.~\ref{fig:GHZProtocolPart2}).
However, the calculation goes through as before because $P_{[\alpha]_{m-1}}^{(n)}=P_{[\alpha]_{m-1}}^{(n)}P_{[1]_m}^{(n)}$, so we can simply insert $P_{[1]_m}^{(n)}$ wherever a measurement has been ``skipped". 
Therefore, putting all the substrings together, the final output state is, up to proportionality and quasi-commuting unitaries, a GHZ state.

Finally, we must account for what happens when we require circuits between measurements. In these cases, the delta functions in Eq.~\eqref{eq:GHZoutput1} become
\begin{align}
&\delta[e=b_{ijk}a_{ijkm}d_{\sigma^{-1}(ijkm)\ominus 1}g_{ijkmL} ]\
\delta[d_{\sigma^{-1}(ijkm)\ominus1}g_{ijkm}=d_{\sigma^{-1}(ijkm)}] \label{eq:GHZOutput5}
\end{align}
where $\sigma$ corresponds to the cumulative permutation implements by the circuits during the protocol. 
However, all this does is rearrange the arguments of the delta functions and the product of those rearranged delta functions is the same object (e.g., $\delta[a=b]\delta[b=c]=\delta[a=c]\delta[c=b]$).
Thus, the outputted state is still a GHZ state up to quasi-commuting unitaries.
Finally, it is clear how the techniques demonstrated in this proof generalize to arbitrary class-$M$ nilpotent groups.

\section{GHZ protocol: Probability distribution of measurement outcomes}\label{Appendix:ProbabilityGHZ}

In this section, we prove that the probability distribution for each measurement outcome during the protocol of Section \ref{sec:GHZProtocol} is uniform, with a probability depending only on the dimension of the corresponding subspace onto which you are projecting.  
As before, it is sufficient to consider a measurement in round $m=3$ having previously measured $x^{(1)}$ and $x^{(2)}$. 
All other rounds can be derived similarly.
Moreover, as before, it is easier to first consider a sequence of outcomes which happen to require no circuits between measurements and then consider the effect of the circuits at the end.
Therefore, we have (compare with Eq.\eqref{eq:SPTEq1})
\begin{align}
p(\vec{x}^{(3)} | \vec{x}^{(1)},\vec{x}^{(2)}) =  \frac{
\begin{aligned}
\includegraphics[width=0.7\linewidth]{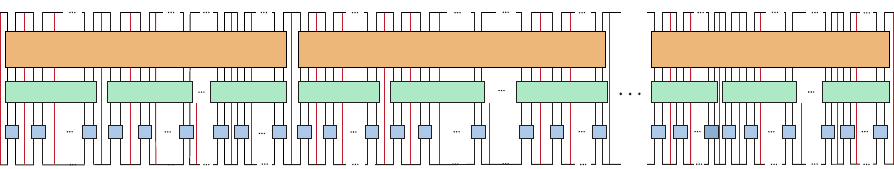}
\end{aligned}
}
{
\begin{aligned}
\includegraphics[width=0.7\linewidth]{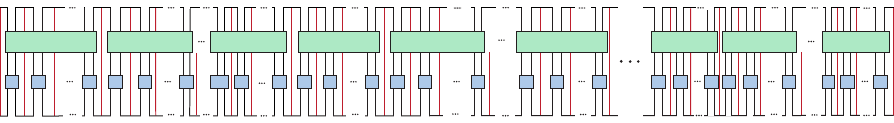}
\end{aligned}} 
\label{eq:GHZEq1}
\end{align}
where the blue boxes correspond to round 1 measurements, green boxes round 2 and orange boxes round 3.
As we assume no circuits are required, we can decompose the measurement outcomes of each round in Appendix \ref{Appendix:VerificationGHZIsOutputed}, e.g., $\vec{x}^{(2)}=(\vec{x}^{(2)}_i)^{N^{(3)}}_{i=1}=((x_{ij}^{(2)})_{j=1}^{N_i^{(2)}})_{i=1}^{N^{(3)}}$.
Let us consider the left-most measurement from round $3$ in the numerator of Eq.~\eqref{eq:GHZEq1}:
\begin{align}
    \begin{aligned}
        \includegraphics[width=0.5\linewidth]{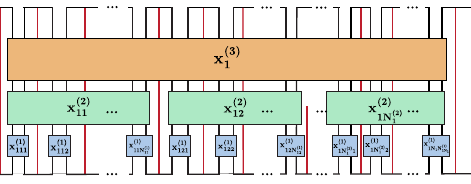}
    \end{aligned}
    \label{eq:ghzprobsegment}
\end{align}
As before, we use Lemma~\ref{lem:IrrepSumTrick} to evaluate this expression. Again, we will have loops of the form
\begin{align}
    \begin{aligned}
        \includegraphics[width=0.22\linewidth]{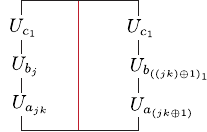}
    \end{aligned}
\end{align}
where, as in Appendix~\ref{Appendix:VerificationGHZIsOutputed}, the $\oplus$ notation indicates the next entry in the multi-index, e.g., $(1,1)\oplus1 =(1,2)$, and the subscript picks out the corresponding sublist, e.g., $(1,2)_1=(1)$.
The physical leg (red line) force this to evaluate as
\begin{equation}
    |G| \delta\left[ c_i b_{ij}a_{ijk}=1\right] \delta\left[ c_i b_{i(jk\oplus1)_1}a_{i(jk\oplus 1)}=1\right],
\end{equation}
where $1$ is the identity element of $G$.
Note, the second delta function will appear in the next loop. 
Thus, Eq.~\eqref{eq:ghzprobsegment} can be evaluated as
\begin{multline}
    \sum_{\{c_1\in G_{M-3}\}}
    \sum_{\{b_{1j}\in G_{M-2}\}_{j}}
    \sum_{\{a_{1jk}\in G_{M-1}\}_{jk}}
    \left(
    \frac{e^{-i \phi^{x_{1}^{(3)}}_{[c_1]}}}{|G_{M-3}|} 
    \left( 
    \prod_{j=1}^{N^{(2)}_{1}} \frac{e^{-i \phi^{x_{1j}^{(2)}}_{[b_{1j}]}}}{|G_{M-2}|} 
    \left( 
    \prod_{k=1}^{N^{(1)}_{1j}} \frac{e^{-i \phi^{x_{1jk}^{(1)}}_{[a_{1jk}]}}}{|G_{M-1}|} 
    \right)
    \right) 
    \right)\\
    \Bigg(
    \prod_{j=1}^{N^{(2)}_{1}}\prod_{k=1}^{N^{(1)}_{1j}}
    |G| \delta[c_1 b_{1j} a_{1jk}=1 ] \Bigg) \sum_{d\in G} \bra{d}\langle c_1 b_{1N_1^{(2)}} a_{1N_1^{(2)}N^{(3)}_{1N_1^{(2)}}} d|  
\end{multline}
Now consider the delta-functions. 
The first delta function is only satisfied if $c_1=a_{111}^{-1}b_{11}^{-1}\in G_{M-2}$.
Therefore $[c_1]=[1]$.
Moreover, we must have $c_1b_{11}a_{111}=1=c_1b_{1j}a_{1jk}$ for all $j,k$.
Thus, $a_{1jk}=a_{1j1}$ and therefore $[a_{1jk}]=[a_{1j1}]$ for all $j$ and $k$ .
Moreover, $b_{1j}=b_{11}a_{111}a_{1j1}^{-1}$ and therefore $[b_{1j}]=[b_{11}]$ for all $j$. As $\circ_k x_{1jk}^{(1)}=[1]$ and $\circ_j x^{(2)}_{1j}=[1]$ for all $j,k$, Eq.~\eqref{eq:ghzprobsegment} therefore ultimately reduces to
\begin{align}
    \frac{|G_{M-2}|}{|G_{M-3}|} 
    \left(\frac{|G_{M-1}|}{|G_{M-2}|}\right)^{N_1^{(2)}} 
    \left(|G|\frac{|G_{M}|}{|G_{M-1}|}\right)^{\sum_{j} N^{(1)}_{1j}}
    \sum_{d\in G} \bra{dd}.
\end{align}
Eq.~\eqref{eq:ghzprobsegment} is just the first term in the numerator of Eq.~\eqref{eq:GHZEq1}, but its clear that the next term will proceed the same way, and so on for each $i$ up until $N^{(3)}$.
Moreover, the denominator will also evaluate the same way but without the $\frac{|G_{M-2}|}{|G_{M-3}|}$ term. 
Thus the overall probability is given by
\begin{align}
    p(x^{(3)}|x^{(2)}x^{(1)})=\frac{|G_{M-2}|}{|G_{M-3}|}^{N^{(3)}}.  
\end{align}
That is, $p(x^{(3)}|x^{(2)}x^{(1)})$ a uniform multinomial distribution, without any parity constraints. 
This same argument will hold for all measurement rounds in Part 1.
Moreover, because $P_{[\alpha]_{m-1}}^{(n)}=P_{[\alpha]_{m-1}}^{(n)}P_{[1]_m}^{(n)}$, it follows that the same holds for the measurements in Part 2.
Finally, if there are circuits between rounds, this only introduces a permutation to indices in the delta functions.
As the delta functions are linked, the global probability distribution is the same. 

\begin{thebibliography}{46}
\bibitem [{not({\natexlab{a}})}]{note1}%
  \BibitemOpen
    \bibinfo {note} {The second cohomology group of a given subgroup, $H\le G$, is a group property of $H$.}\BibitemShut {Stop}%
\bibitem [{not({\natexlab{b}})}]{note2}%
  \BibitemOpen
    \bibinfo {note} {Combination here is achieved by introducing an auxiliary site and using a low-depth circuit consisting of SWAP gates to exchange the auxiliary system with one of the measured sites.}\BibitemShut {Stop}%
\bibitem [{not({\natexlab{c}})}]{note4}%
  \BibitemOpen
    \bibinfo {note} {The reason for that is that if the outcomes would multiply to any other equivalence class, then the projector onto it would be orthogonal to the input state. Hence, the outcomes on the first $N-1$ sites determine the outcome of the last measurement. Note that a simple example of such a parity constraint is given in \cite{Gunn2025}.}\BibitemShut {Stop}%
\bibitem [{not({\natexlab{d}})}]{note5}%
  \BibitemOpen
    \bibinfo {note} {As an example, consider the string $(1,0,\dots,0,1)$. If sufficiently large, this outcome contains an $L^{(1)}$ substring, but due to periodicity the first and last site are nearest neighbors, and thus can be paired using a single SWAP gate.}\BibitemShut {Stop}%
\bibitem [{not({\natexlab{e}})}]{note6}%
  \BibitemOpen
    \bibinfo {note} {We use italics label to indicate conjugacy classes. Each conjugacy class is a set of elements. The number of the label indicates the order of the elements in the class and the letter labels the different conjugacy classes with the same order.}\BibitemShut {Stop}%
\end{thebibliography}
\end{document}